\documentclass[10pt,conference]{IEEEtran}
\usepackage{amsmath}
\usepackage{algorithmic}
\usepackage{algorithm}
\usepackage{stfloats}
\usepackage{graphicx}

\newtheorem{theorem}{Theorem}
\newtheorem{lemma}{Lemma}
\newtheorem{corollary}{Corollary}

\begin{document}

\title{Generating Probability Distributions using Multivalued Stochastic Relay Circuits}

\author{\authorblockN{David Lee}
\authorblockA{Dept. of Electrical Engineering\\
Stanford University\\
Palo Alto, CA 94305\\
Email: davidtlee@stanford.edu}
\and
\authorblockN{Jehoshua Bruck}
\authorblockA{Dept. of Electrical Engineering\\
California Institute of Technology\\
Pasadena, CA 91125\\
Email: bruck@caltech.edu}}

\maketitle

\begin{abstract}
The problem of random number generation dates back to von Neumann's work in 1951. Since then, many algorithms have been developed for generating unbiased bits from complex correlated sources as well as for generating arbitrary distributions from unbiased bits. An equally interesting, but less studied aspect is the {\it structural} component of random number generation as opposed to the algorithmic aspect. That is, given a network structure imposed by nature or physical devices, how can we build networks that generate arbitrary probability distributions in an optimal way?

In this paper, we study the generation of arbitrary probability distributions in multivalued relay circuits, a generalization in which relays can take on any of $N$ states and the logical `and' and `or' are replaced with `min' and `max' respectively. Previous work was done on two-state relays. We generalize these results, describing a duality property and networks that generate arbitrary rational probability distributions. We prove that these networks are robust to errors and design a universal probability generator which takes input bits and outputs arbitrary binary probability distributions.
\end{abstract}

\section{Introduction}

\subsection{Motivation}
Many biological systems involve stochasticity. Examples of these include gene expression\cite{gene}, chemical reactions\cite{crn}, and neuron signaling\cite{neuron}. However, despite the stochasticity, often deemed as noise, they are still capable of achieving functionalities that artificial systems cannot yet compete with. 

Motivated by the idea that stochasticity is an important enabler of biological computation, we tackle a simpler question as a stepping stone: Can we design networks that generate stochasticity in a systematic way?

\subsection{Structural Aspects of Random Number Generation}
This question is strongly connected to an important thread of work in computer science on random number generation. In 1951, Von Neumann \cite{Neumann1951} studied this problem in the context of generating a fair coin toss from a biased coin. Knuth and Yao \cite{Knuth1976} studied the reverse problem of generating arbitrary distributions from unbiased bits (fair coins), which was extended by Han and Hoshi \cite{Han1997} to generating arbitrary distributions from biased distributions. These primarily focus on an algorithmic perspective to random number generation.

In a biological system (or other physical device), we do not have the full generality of the algorithmic approach to generating probability distributions. Randomness arises in specific areas that can only propagate according to the structure of how components are composed in the given network. The study of random number generation under these constraints can be greatly beneficial to understanding natural stochastic systems as well as to designing devices for random number generation. In this paper, we will analyze random number generation in the context of multivalued relay circuits, a generalization of standard relay circuits to any number of states. This generalization is very natural and has an intuitive understanding as the timing of events that have dependencies on other events.

\subsection{Deterministic Relays}

\begin{figure}[!b]
\centering
\includegraphics[width=3.0in]{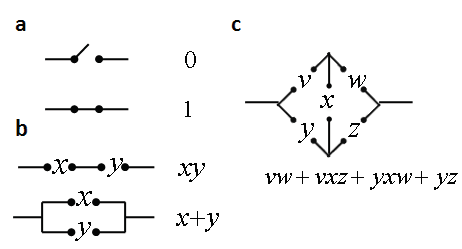}
\caption{\textbf{Deterministic Examples.} For two states, $xy$ is the boolean `and' of $x$ and $y$ while $x+y$ is the boolean `or'. Figures b and c also apply to multivalued relays where $xy$ is $min(x, y)$ and $x+y$ is $max(x, y)$. \textit{(a)} A relay can either be opened (state 0) or closed (state 1). \textit{(b)} In series, the entire circuit can only be closed if both $x$ and $y$ are closed. In parallel, the entire circuit will be closed if either $x$ or $y$ are closed. \textit{(c)} In this example of a non-sp circuit, we take the boolean `and' along the four possible paths, and then apply the boolean `or' to those subsequent values.}
\label{det_example}
\end{figure}

We will start by introducing deterministic relays. A deterministic relay switch is a 2 terminal object which can be connected (closed) by a wire or left open. The state of the switch, which can be either 0 or 1, describes an open or closed state respectively (see Figure \ref{det_example}a). These states are complements of each other. That is, if switch $x$ is in state 0, then the complement $\bar{x}$ is in state 1 and vice versa. When multiple switches are composed together, these networks are known as relay circuits. One of these, a series composition, is formed when two switches are put end to end. A parallel composition is formed when two switches are composed so that the beginning terminals are connected together and the end terminals are connected together.

Shannon showed that the series and parallel compositions can be represented by the boolean `and' and `or' operations\cite{Shannon}. If switches $x$ and $y$ are composed in series to form $z_1$, then $z_1$ will only be closed if both $x$ and $y$ are closed. On the other hand, if switches $x$ and $y$ are composed in parallel to form $z_2$, then $z_2$ will be closed if either $x$ or $y$ are closed. We will denote the series composition of $x$ and $y$ by $x*y$ or simply $xy$ and the parallel composition by $x+y$ (see Figure \ref{det_example}b). This notation will be preserved for further generalizations of the relay circuits.

Circuits formed solely by series and parallel compositions are called sp circuits. Shannon showed that non sp circuits could also be represented by boolean operations. The general mathematical representation of any relay circuits would be to find all paths that go from the beginning to the end terminal. For each path, we take the boolean `and' of all switches along that path; then we take the boolean `or' of the values derived for each path (see Figure \ref{det_example}c). 

\subsection{Stochastic Relays}

\begin{figure}[!b]
\centering
\includegraphics[width=3.0in]{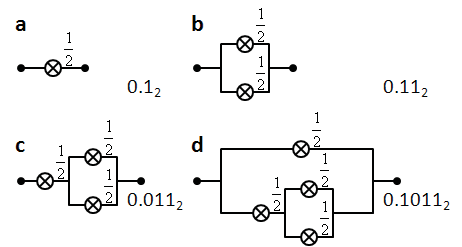}
\caption{\textbf{Simple Examples.} As stated, $\frac{1}{2}$ is short for $(\frac{1}{2}, \frac{1}{2})$. \textit{(a)} A single $\frac{1}{2}$-pswitch. \textit{(b)} Putting (a) in parallel with a $\frac{1}{2}$-pswitch gives $\frac{3}{4}$. \textit{(c)} Putting (b) in series with a $\frac{1}{2}$-pswitch gives $\frac{3}{8}$. \textit{(d)} Putting (c) in parallel with a $\frac{1}{2}$-pswitch gives $\frac{11}{16}$.}
\label{example}
\end{figure}

Recently, Wilhelm and Bruck introduced the notion of a stochastic relay circuit\cite{WilhelmBruck}. These circuits are a generalization of Shannon's relay circuits; instead of having deterministic relay switches that are in either the open or closed state, stochastic relay circuits can exist in either state with a specified probability. If a stochastic relay switch $x$, called a pswitch, has probability $p$ of being closed and probability $1-p$ of being open, this {\it distribution} is represented by a vector $v = (1-p, p)$ where $v_i$ corresponds to the probability of $x$ being in state $i$. We say that $x$ {\it realizes} $(1-p, p)$ or simply $x$ {\it realizes} $p$. If pswitches $x$ and $y$, which realize probabilities $p$ and $q$ respectively, are composed in series, the new composition will realize $pq$. If they are composed in parallel, the new composition will realize $p + q - pq$ (see Figure \ref{example}). 

One of the primary questions dealt with in their work was the generation of probability distributions using a limited number of base switch types, known as a switch set. For example, if the switch set {\bf S} $ = \{\frac{1}{2}\}$, then relay circuits built with this switch set can only use pswitches with the distribution $(\frac{1}{2}, \frac{1}{2})$. They proved that using the switch set {\bf S} $ = \{\frac{1}{2}\}$, all probability distributions $\frac{a}{2^n}$ could be realized with at most $n$ pswitches. Continuing, many more results were proved not only in realizing other probability distributions\cite{WilhelmBruck}\cite{ZhouBruck}, but also in circuit constructions such as a Universal Probability Generator\cite{WilhelmBruck} and in robustness\cite{LohZhouBruck} and expressibility\cite{ZhouBruck} properties of these circuits.

\subsection{Multivalued Stochastic Relays}

In order to study non-bernoulli random variables, it is necessary to generalize Shannon's relays to a larger number of states. Multivalued logics have been studied as early as in 1921 by Post\cite{Post} and followed up on by Webb\cite{Webb} and others\cite{Rine}. The work presented in this paper concerns one generalization of two-state relay circuits to multivalued relay circuits where we use two multivalued functions (gates).

A multivalued switch is a relay switch that can be in any of n states: 0, 1, 2, ..., n-1. We define the complement of a switch to be $n-1-i$, where $i$ is the state of the switch. Series and parallel compositions are redefined to `min' and `max', respectively, rather than the boolean `and' and `or'. This means that when switches $x$ and $y$ are composed in series, the overall circuit is in state $\min(x, y)$ and when they are composed in parallel, the overall circuit is in state $\max(x, y)$ (see Figure \ref{multi_example}a and b). Non-sp circuits are also defined in a similar way to 2-state circuits. The general mathematical representation of any multivalued relay circuit is to find all paths that go from the beginning to the end terminal. For each path, we take the `min' of all switches along that path; then we take the `max' of the values derived for each path (Figures \ref{det_example}b and c still apply). 

\begin{figure}[!b]
\centering
\includegraphics[width=3.0in]{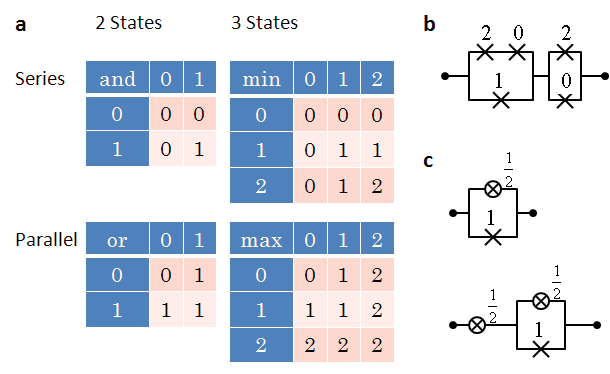}
\caption{\textbf{3-state Examples.} We use $\times$ to represent a deterministic switch and $\otimes$ to represent a pswitch. \textit{(a)} A comparison of the truth tables for 2-state and 3-state relay switches. \textit{(b)} A simple 3-state deterministic example. Evaluating series connections using min and parallel connections as max, we find that this circuit is in state 1. \textit{(c)} A simple 3-state stochastic example. On the top, we put a $(\frac{1}{2}, 0, \frac{1}{2})$ pswitch in parallel with a deterministic switch in state 1. Then, with $\frac{1}{2}$ probability we get $\max(0, 1)$ and with $\frac{1}{2}$ probability we get $\max(2, 1)$, which yields a $(0, \frac{1}{2}, \frac{1}{2})$ circuit. Similarly, we can calculate the bottom circuit to realize distribution $(\frac{1}{2}, \frac{1}{4}, \frac{1}{4})$.}
\label{multi_example}
\end{figure}

One physical understanding of this max-min algebra is found in the timing of relay switches. Let a switch start off in the closed position. Then the state of the switch will be the time $t \in \{0, 1, ..., n-1\}$ when the switch opens. Compose two switches, which open at time $t_1$ and $t_2$, in series. If either switch opens, then the overall connection is broken. Therefore, the time that the series composition will be open is $\min(t_1, t_2)$. In the same way, we can compose the switches in parallel. Since both switches need to be opened in order for the overall circuit to be open, the time calculated for the parallel composition is $\max(t_1, t_2)$. More generally, the max-min algebra can be used to understand the timing of events when an event occurence depends on the occurence of other events, i.e. event $A$ occurs if event $B$ occurs or if event $C$ and $D$ occur. Max-min functions have been studied as a way of reasoning about discrete event systems\cite{Gunawardena1994}.

Now that we have defined multivalued relays, we can also reason about probability distributions over the $n$ states of a multivalued stochastic relay. If a 3-state pswitch $x$ has probability $p_0$ of being in state 0, $p_1$ of being in state 1, and $p_2$ of being in state 2, we represent the distribution with the vector $v = (p_0, p_1, p_2)$ and say $x$ realizes $(p_0, p_1, p_2)$. If the distribution is in the form $(1-p, 0, ..., 0, p)$, we will shorten this to simply $p$ if the number of states can be inferred from the context or if the equation holds for any number of states (see Figure \ref{multi_example}c).

We present the following results in our paper:
\begin{enumerate}
\item A duality property (Section II).
\item Networks for generating binary probability distributions (Sections III and IV).
\item Networks for generating rational distributions (Section V).
\item Robustness of the previous networks to switch errors (Section VI).
\item Universal Probability Generation (Section VII).
\item Switching with partially ordered states (Section VIII).
\item Applications to neural circuits and DNA Computing (Section IX).
\end{enumerate}

\section{Duality}

It is important to characterize properties of multivalued circuits. One well-known property is duality, which plays a role in resistor networks, deterministic and two-state stochastic relay circuits. We show that a similar duality concept exists for multivalued circuits.
Define the dual state of $i$ as the state $n-1-i$; the dual distribution of $v$ as the distribution $\bar{v}$ where $\bar{v}_i = v_{n-1-i}$; the dual switch of $x$ as the switch $\bar{x}$ that realizes the dual distribution of $x$; and the dual circuit of $C$ as the circuit $\bar{C}$  that realizes the dual distribution of $C$.

\begin{figure}[!t]
\centering
\includegraphics[width=3.0in]{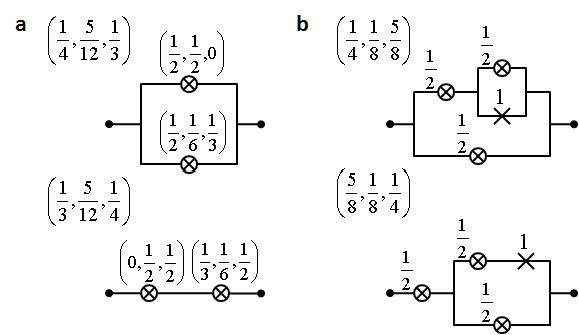}
\caption{\textbf{3-state Duality Example.} Remember that the 3-state $\frac{1}{2}$-pswitch is the shorthand for the $(\frac{1}{2}, 0, \frac{1}{2})$ pswitch. \textit{(a)} A two element duality example. \textit{(b)} A duality example on a larger circuit. Note that since the pswitches are symmetric distributions (duals of themselves), only the compositions are changed.}
\label{duality}
\end{figure}

\begin{theorem}[Duality Theorem]
Given a stochastic series-parallel circuit $C$, we can construct $\bar{C}$ by replacing all the switches in $C$ with their dual switches and by replacing series connections with parallel connections and vice versa. (see Figure \ref{duality}).
\end{theorem}
\begin{proof}
This is shown using induction on series-parallel connections. 

Base Case: The dual of a single pswitch with distribution $(p_{0}, p_{1}, ..., p_{n-1})$ is $(p_{n-1}, ..., p_{1}, p_{0})$, which trivially satisfies the theorem. 

Inductive Step: Suppose a circuit $C$ with distribution $(p_{0}, p_{1}, ..., p_{n-1})$ and a circuit $C'$ with distribution $(q_{0}, q_{1}, ..., q_{n-1})$ satisfy the theorem, i.e. the distribution of $\bar{C}$ is $(p_{n-1}, ..., p_{1}, p_{0})$ and the distribution of $\bar{C'}$ is $(q_{n-1}, ..., q_{1}, q_{0})$. To prove the theorem, it is sufficient for us to show that $CC'$ is the dual of $\bar{C} + \bar{C'}$.

Let $C_s = CC'$ and $C_p = \bar{C} + \bar{C'}$. Then $C_{s} = (c_{0}, c_{1}, ..., c_{n-1})$ and $C_{p} = (d_{0}, d_{1}, ..., d_{n-1})$ where
\begin{align*}
\displaystyle c_{k} &= \sum\limits_{min(i,j)=k} Pr(C=i)Pr(C'=j)\\
                    &= \sum\limits_{min(i,j)=k} p_{i}q_{j}\\
              d_{k} &= \sum\limits_{max(i,j)=k} Pr(\bar{C}=i)Pr(\bar{C'}=j)\\
                    &= \sum\limits_{max(i,j)=k} p_{n-1-i}q_{n-1-j}\\
                    &= \sum\limits_{min(n-1-i,n-1-j)=n-1-k} p_{n-1-i}q_{n-1-j}\\
                    &= \sum\limits_{min(i',j')=n-1-k} p_{i'}q_{j'}\\
                    &= c_{n-1-k}.
\end{align*}
We see that $d_{0} = c_{n-1}$, $d_{1} = c_{n-2}$, ..., $d_{n-1} = c_{0}$, demonstrating that $C_s$ is the dual of $C_p$.
\end{proof}

\section{Realizing Binary 3-state Distributions}

We can now ask questions about generating probability distributions with stochastic relay circuits. These include: What are the possible distributions that can be realized? What is the smallest set of basic switching elements necessary to realize these distributions? Are there efficient constructions to realize these distributions? We will begin by demonstrating how to generate distributions on 3 states.

In writing out algorithms and circuit constructions we will use the notation introduced earlier for series and parallel connections, i.e. $xy$ and $x+y$ will denote series and parallel connections respectively. We will use this notation loosely, i.e., $(\frac{1}{2}, 0, \frac{1}{2})(\frac{1}{2}, \frac{1}{2}, 0)$ represents a circuit formed by composing a $(\frac{1}{2}, 0, \frac{1}{2})$ pswitch in series with a $(\frac{1}{2}, \frac{1}{2}, 0)$ pswitch. 

\begin{lemma}[3-state sp composition rules]
Given $p = (p_0, p_1, p_2)$ and $q = (q_0, q_1, q_2)$, let $x = pq$ and $y = p + q$. Then,
\begin{align*}
\displaystyle x_0 &= p_0q_0 + p_0q_1 + p_0q_2 + p_1q_0 + p_2q_0\\
                  &= 1 - (1-p_0)(1-q_0)\\
              x_1 &= p_1q_1 + p_1q_2 + p_2q_1\\
              x_2 &= p_2q_2\\
              y_0 &= p_0q_0\\
              y_1 &= p_1q_1 + p_1q_0 + p_0q_1\\
              y_2 &= p_2q_2 + p_2q_1 + p_2q_0 + p_1q_2 + p_0q_2\\
                  &= 1 - (1-p_2)(1-q_2)
\end{align*}
\end{lemma}
\begin{proof}
The above expressions follow from enumerating all $3^2$ switch combinations.
\end{proof}

\begin{figure}[!t]
\centering
\includegraphics[width=3.0in]{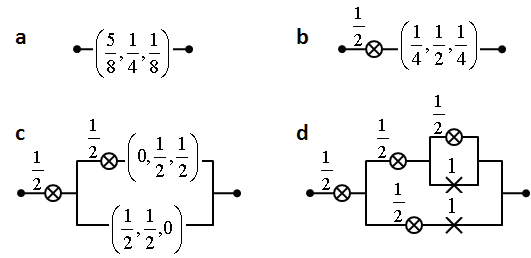}
\caption{\textbf{Realizing 3-state Distribution.} \textit{(a)} We want to realize $(\frac{5}{8}, \frac{1}{4}, \frac{1}{8})$. This falls under the first case. \textit{(b)} Now we have an unconstructed distribution that falls under the second case. \textit{(c)} There are two unconstructed distributions. The top one falls under case 3 and the bottom falls under case 2. (Note that case 2 is equivalent to case 1 for $p_0 = \frac{1}{2}$). \textit{(d)} The final circuit.}
\label{algthm_3}
\end{figure}

\begin{theorem}[Binary 3-state Distributions]
Using the switch set {\bf S} $ = \{\frac{1}{2}\}$ and the deterministic switches, we can realize all 3-state distributions of the form $(\frac{a}{2^n}, \frac{b}{2^n}, \frac{c}{2^n})$ using at most $2n-1$ pswitches with the following recursive circuit construction (see Figure \ref{algthm_3} for an example):

\begin{algorithm}[H]
\begin{algorithmic}
\IF {$p \in S$}
		\STATE return $p$
\ELSIF {$\frac{a}{2^n} > \frac{1}{2}$}
    \STATE return $(\frac{1}{2}, 0, \frac{1}{2})(\frac{a - 2^{n-1}}{2^{n-1}}, \frac{b}{2^{n-1}}, \frac{c}{2^{n-1}})$
\ELSIF {$\frac{a+b}{2^n} > \frac{1}{2}$}
    \STATE return $(\frac{a}{2^{n-1}}, \frac{2^{n-1} - a}{2^{n-1}}, 0) + (\frac{1}{2}, 0, \frac{1}{2})(0, \frac{2^{n-1}-c}{2^{n-1}}, \frac{c}{2^{n-1}})$
\ELSIF {$\frac{a + b + c}{2^n} > \frac{1}{2}$}
    \STATE return $(\frac{a}{2^{n-1}}, \frac{b}{2^{n-1}}, \frac{c - 2^{n-1}}{2^{n-1}}) + (\frac{1}{2}, 0, \frac{1}{2})$
\ENDIF
\end{algorithmic}
\end{algorithm}
\end{theorem}

\begin{proof}
For any distribution $p = (p_0, p_1, p_2)$, there exists some smallest $k$ such that $\sum_{i=0}^{k}p_i > \frac{1}{2}$, which correspond to the 3 recursive cases enumerated above. We can verify that for each of these cases:
\begin{enumerate}
\item The decompositions obey the 3-state composition rules.
\item The switches are valid (non-negative probabilities and sum to 1)
\end{enumerate}
Since each algorithm's decomposition uses switches of type $(\frac{a}{2^{n-1}}, \frac{b}{2^{n-1}}, \frac{c}{2^{n-1}})$, then we will eventually have $n=0$, corresponding to a deterministic switch; at this point the algorithm terminates and has successfully constructed any $(\frac{a}{2^n}, \frac{b}{2^n}, \frac{c}{2^n})$.

\begin{figure}[!t]
\centering
\includegraphics[width=3.0in]{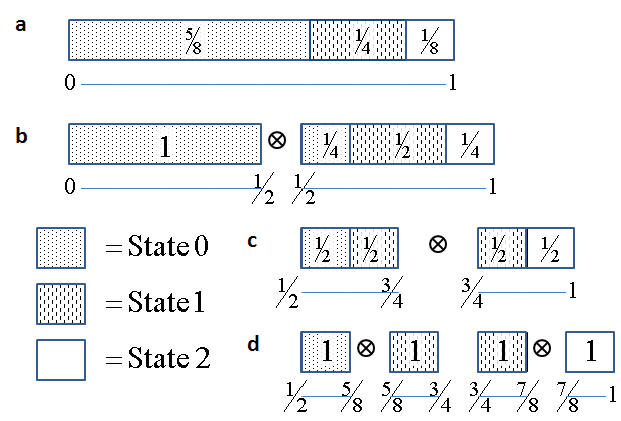}
\caption{\textbf{Block-Interval visualization of 3-state Algorithm.} We want to realize $(\frac{5}{8}, \frac{1}{4}, \frac{1}{8})$. A $\otimes$ symbol marks where the cuts occur which correspond to the use of one pswitch. \textit{(a)} We start with one interval split into blocks corresponding to the pswitch probabilities. \textit{(b)} When we cut $[0, 1]$ in half, we get two intervals; the blocks are cut and the probabilities change according to the new interval size. \textit{(c)} The second application of the algorithm. \textit{(d)} When we are left with single block-interval pairs, we know we are done since the deterministic switch is in our switch set.}
\label{visualization}
\end{figure}

We will now prove that we use at most $2n-1$ pswitches for all $n \geq 1$. Define $f_n$ as the maximum number of pswitches used in the construction of any distribution $(\frac{a}{2^n}, \frac{b}{2^n}, \frac{c}{2^n})$. Then,
\begin{align}
\displaystyle f_1 &= 1\label{algthm_3_base}\\
              f_n &= \max(f_{n-1} + 1, 2(n-1) + 1, f_{n-1} + 1)\label{algthm_3_recur}\\
                  &= \max(f_{n-1} + 1, 2n-1)\notag
\end{align} 
\noindent where (\ref{algthm_3_base}) is shown trivially, and (\ref{algthm_3_recur}) is derived from the 3 cases of the algorithm. Note that we are using a previous result\cite{WilhelmBruck} that 2 state distributions of form $(\frac{a}{2^n}, \frac{b}{2^n})$ use at most $n$ switches, and also assuming here that our algorithm generalizes the previous algorithm; that is, distributions of the form $(\frac{a}{2^n}, \frac{b}{2^n}, 0)$ and its permutations also use at most $n$ switches -- this application of their 2-state proof to our multivalued distribution is sound as will be explained more rigorously in the following section. We are now left with a simple induction exercise.

Base Case: $f_1 = 1 = 2(1)-1$.

Inductive Step: Assume $f_k = 2k-1$. Then,
\begin{align*}
f_{k+1} &= \max(f_k + 1, 2(k+1)-1)\\
        &= \max(2k, 2(k+1)-1)\\
        &= 2(k+1)-1. 
\end{align*}
So $f_k = 2k-1 \implies f_{k+1} = 2(k+1)-1$.
\end{proof}

It is useful at this time to provide some intuition regarding the algorithm. We can view the original distribution as a series of blocks dividing up the interval $[0, 1]$ (see Figure \ref{visualization}). By applying the algorithm, we are separating this larger interval into smaller intervals $[0, \frac{1}{2}]$ and $[\frac{1}{2}, 1]$, cutting any blocks on that boundary. When this separation occurs, the total size of the block is decreased (namely, in half), and so the probabilities representing those intervals change - these probabilities are precisely those of the algorithm. Namely, we can rewrite the three cases of the algorithm:
\begin{align*}
\textstyle{(\frac{a}{2^n}, \frac{b}{2^n}, \frac{c}{2^n})} = \\
\text{Case 1: } &\textstyle{(1, 0, 0) + (\frac{1}{2}, 0, \frac{1}{2})(\frac{a - 2^{n-1}}{2^{n-1}}, \frac{b}{2^{n-1}}, \frac{c}{2^{n-1}})}\\
\text{Case 2: } &\textstyle{(\frac{a}{2^{n-1}}, \frac{2^{n-1} - a}{2^{n-1}}, 0) + (\frac{1}{2}, 0, \frac{1}{2})(0, \frac{2^{n-1}-c}{2^{n-1}}, \frac{c}{2^{n-1}})}\\
\text{Case 3: } &\textstyle{(\frac{a}{2^{n-1}}, \frac{b}{2^{n-1}}, \frac{c - 2^{n-1}}{2^{n-1}}) + (\frac{1}{2}, 0, \frac{1}{2})(0, 0, 1)}
\end{align*}

\section{Realizing Binary $N$-state Distributions}

We are now ready to continue with our algorithm for $N$ states. Intuitively, we can describe the algorithm for $N$ states in the same way as for 3 states. We first find the smallest index $k$ for which $\sum_{i=0}^{k}p_i > \frac{1}{2}$. Then based on the index $k$, we can decompose our distribution in a way corresponding to the interval-block visualization; the only difference is that each interval can now have up to $N$ block-types instead of just 3.

\begin{figure}[!t]
\centering
\includegraphics[width=3.0in]{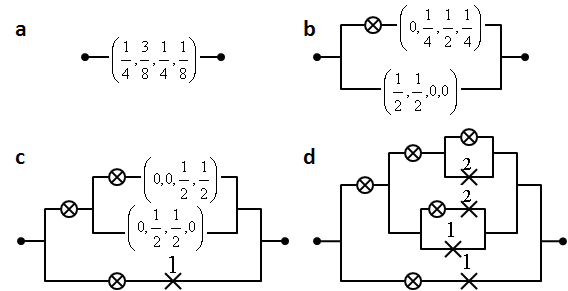}
\caption{\textbf{Realizing 4-state Distribution.} We use the $N$-state algorithm on the a 4-state distribution. A $\otimes$ symbol represents the pswitch $(\frac{1}{2}, 0, 0, \frac{1}{2})$. \textit{(a)} We want to realize $(\frac{1}{4}, \frac{3}{8}, \frac{1}{4}, \frac{1}{8})$. Initially, $p_0 + p_1 > \frac{1}{2}$. \textit{(b)} For the top distribution, $p_0 + p_1 + p_2 > \frac{1}{2}$. For the bottom distribution, $p_0 + p_1 > \frac{1}{2}$. \textit{(c)} For the top distribution, $p_0 + p_1 + p_2 + p_3 > \frac{1}{2}$. For the bottom distribution, $p_0 + p_1 + p_2 > \frac{1}{2}$. \textit{(d)} The final circuit.}
\label{algthm_n}
\end{figure}

\begin{theorem}[Binary N-state Distributions]
Using the switch set {\bf S} $ = \{\frac{1}{2}\}$ and the deterministic switches, we can realize all $N$-state distributions of the form $(\frac{x_0}{2^n}, \frac{x_1}{2^n}, ..., \frac{x_{N-1}}{2^n})$ with at most $f_{n, N} =$
\begin{align*}
\begin{cases} 
2^n-1, & n \leq \lceil\log_{2}N\rceil\\
2^{\lceil\log_{2}N\rceil}-1 + (N-1)(n-\lceil\log_{2}N\rceil), & n \geq \lceil\log_{2}N\rceil
\end{cases}
\end{align*}
switches, using the following recursive circuit construction (see Figure \ref{algthm_n}):

\begin{algorithm}[H]
\caption{Note that we are using $\frac{1}{2}$ as a shorthand for $(\frac{1}{2}, ..., \frac{1}{2})$ as explained in the introduction.}
\begin{algorithmic}
\IF {$p \in S$}
		\STATE return $p$
\ELSIF {$\frac{x_0}{2^n} > \frac{1}{2}$}
    \STATE return $\frac{1}{2}(\frac{x_0 - 2^{n-1}}{2^{n-1}}, \frac{x_1}{2^{n-1}}, ..., \frac{x_{N-1}}{2^{n-1}})$
\ELSIF {$\frac{x_0+x_1}{2^n} > \frac{1}{2}$}
    \STATE return $(\frac{x_0}{2^{n-1}}, \frac{2^{n-1} - x_0}{2^{n-1}}, ...) + \frac{1}{2}(0, \frac{x_0 + x_1 - 2^{n-1}}{2^{n-1}}, \frac{x_2}{2^{n-1}}, ...)$
\ELSIF {...}
    \STATE ...
\ELSIF {$\frac{x_0 + x_1 + ... + x_{N-1}}{2^n} > \frac{1}{2}$}
    \STATE return $(\frac{x_0}{2^{n-1}}, \frac{x_1}{2^{n-1}}, ..., \frac{x_{N-1} - 2^{n-1}}{2^{n-1}}) + \frac{1}{2}$
\ENDIF
\end{algorithmic}
\end{algorithm}
\end{theorem}

The algorithm's correctness can be confirmed from the composition rules for N-states. Before proving the complexity of the above algorithm, we state an intermediate lemma:
\begin{lemma}
Define an active state of a distribution as a state with non-zero probability. Given any distribution of the form $(\frac{x_0}{2^n}, \frac{x_1}{2^n}, ..., \frac{x_{N-1}}{2^n})$ with $m \leq N$ active states, the number of pswitches needed to realize the distribution is at most $f_{n, m}$.
\end{lemma}
\begin{proof}
Re-map the $m$ active states to $0, 1, ..., m-1$ in a way that preserves state order. Since all operations under $\max$ and $\min$ are preserved, any distribution on the original $m$ active states can be constructed by constructing a distribution on the mapped states and reversing the mapping on the base switches. Therefore, the number of pswitches required cannot be greater than $f_{n, m}$. This also gives us insight into $f_{n, N}$ for the region $n \leq \lceil\log_{2}N\rceil$. In this region, the number of active states is limited by $n$, resulting in states with zero probability. Then, it makes sense that the complexity only depends on $n$.
\end{proof}

\begin{proof}
We will now prove the complexity of pswitches required. The following recursive relations hold for $f_{n,N}$,
\begin{align}
&f_{n, 1} &&= 0\label{complexity_3}\\
&f_{0, N} &&= 0\label{complexity_4}\\
&f_{n, N} &&= \max(f_{n-1, 1} + f_{n-1, N} + 1,\label{complexity_5}\\
&         &&f_{n-1, 2} + f_{n-1, N-1} +1,\notag\\
&         &&...,\notag\\
&         &&f_{n-1, N} + f_{n-1, 1} + 1)\notag\\
&         &&= \max_{1 \leq i \leq \lceil\frac{N}{2}\rceil}(f_{n-1, i} + f_{n-1, N-i+1}) + 1\label{complexity_eqn}
\end{align}

where (\ref{complexity_3}) and (\ref{complexity_4}) are trivial since both define deterministic switches and (\ref{complexity_5}) comes from an application of the above lemma to the different cases of the algorithm. If we fill out a table based on the recursive relations (see Table I), we find that the values match our closed form hypothesis. To prove it, we will:
\begin{enumerate}
\item Demonstrate that (\ref{complexity_eqn}) is maximized at $i = \lceil\frac{N}{2}\rceil$.
\item Use induction to prove $f_{n, N}$ for $n \leq \lceil\log_{2}N\rceil$.
\item Use induction to prove $f_{n, N}$ for $n \geq \lceil\log_{2}N\rceil$.
\end{enumerate}

\begin{table}
\begin{tabular}{c|ccccccccccc}
    & 0 & 1 & 2 & 3 & 4  & 5  & 6  & 7  & 8  & 9  & 10\\
\hline			
  1 & {\bf 0} & 0 & 0 & 0 & 0  & 0  & 0  & 0  & 0  & 0  & 0\\
  2 & 0 & {\bf 1} & 2 & 3 & 4  & 5  & 6  & 7  & 8  & 9  & 10\\
  3 & 0 & 1 & {\bf 3} & 5 & 7  & 9  & 11 & 13 & 15 & 17 & 19\\
  4 & 0 & 1 & {\bf 3} & 6 & 9  & 12 & 15 & 18 & 21 & 24 & 27\\
  5 & 0 & 1 & 3 & {\bf 7} & 11 & 15 & 19 & 23 & 27 & 31 & 35\\
  6 & 0 & 1 & 3 & {\bf 7} & 12 & 17 & 22 & 27 & 32 & 37 & 42\\
  7 & 0 & 1 & 3 & {\bf 7} & 13 & 19 & 25 & 31 & 37 & 43 & 49\\
  8 & 0 & 1 & 3 & {\bf 7} & 14 & 21 & 28 & 35 & 42 & 49 & 56\\
  9 & 0 & 1 & 3 & 7 & {\bf 15} & 23 & 31 & 39 & 47 & 55 & 63\\
\end{tabular}
\label{complexity_table}
\caption{(row, column) values correspond to ($N$, $n$). The $n = \lceil\log_{2}N\rceil$ border is bolded.}
\end{table}

\noindent{\it Part 1:} Define $\Delta_{n,N} = f_{n, N+1} - f_{n, N}$.

\noindent If $n \geq \lceil\log_{2}(N+1)\rceil$, 
\begin{align*}
\Delta_{n, N} &= [(n-\lceil\log_{2}(N+1)\rceil)(N+1-1) + 2^{\lceil\log_{2}(N+1)\rceil}]\\
              &\ \ - [(n-\lceil\log_{2}N\rceil)(N-1) + 2^{\lceil\log_{2}N\rceil}]
\end{align*}
For $N$ a multiple of 2, $\lceil\log_{2}(N+1)\rceil = \lceil\log_{2}N\rceil+1$, so
\begin{align*}
\Delta_{n, N} &= [(n-\lceil\log_{2}N\rceil-1)N + 2^{\lceil\log_{2}N\rceil+1}]\\
              &\ \ - [(n-\lceil\log_{2}N\rceil)(N-1) + 2^{\lceil\log_{2}N\rceil}]\\
              &= -N + n - \lceil\log_{2}N\rceil + N\\
              &= n - \lceil\log_{2}N\rceil
\end{align*}
For $N$ not a multiple of 2, $\lceil\log_{2}(N+1)\rceil = \lceil\log_{2}N\rceil$, so
\begin{align*}
\Delta_{n, N} &= n - \lceil\log_{2}N\rceil
\end{align*}

\noindent If $n \leq \lceil\log_{2}N\rceil$,
\begin{align*}
\Delta_{n, N} &= 2^n - 2^n = 0
\end{align*}

\noindent So we have,
\begin{align*}
\Delta_{n, N} &=
\begin{cases} 
n-\lceil\log_{2}N\rceil, &n \geq \lceil\log_{2}(N+1)\rceil\\
0, &n \leq \lceil\log_{2}N\rceil
\end{cases}
\end{align*}
\noindent Since $\Delta_{n, N}$ is nonincreasing for a fixed $n$ and increasing $N$, the maximum of $f_{n-1, i} + f_{n-1, N-i+1} +1$ is at the highest $i$ for which $i \leq N-i+1$. Otherwise, we could increment $i$, increasing the entire quantity by $\Delta_{n-1, i} - \Delta_{n-1, N-i} \geq 0$. Therefore the maximizing index is $i = \lceil\frac{N}{2}\rceil$.\\

\noindent{\it Part 2:} Prove $f_{n,N}$ for $n \leq \lceil\log_{2}N\rceil$.

\noindent Base Case: $f(0, N) = 0 = 2^0 - 1$.

\noindent Inductive Step: Assume $f_{k-1, N'} = 2^{k-1}-1$ $\forall N' \leq N$ and $k-1 \leq \lceil\log_{2}N'\rceil$. Then for $k \leq \lceil\log_{2}N\rceil$,
\begin{align} 
f_{k, N} &= f_{k-1, \lceil\frac{N}{2}\rceil} + f_{k-1, N-\lceil\frac{N}{2}\rceil+1} + 1\notag\\
         &= f_{k-1, \lceil\frac{N}{2}\rceil} + f_{k-1, \lfloor\frac{N}{2}\rfloor+1} + 1\notag\\
         &= 2(2^{k-1}-1)+1\label{complexity_8}\\
         &= 2^k+1\notag
\end{align}
\noindent where (\ref{complexity_8}) is true because
\begin{align*}
k \leq \lceil\log_{2}N\rceil &\implies k-1 \leq \lceil\log_{2}\frac{N}{2}\rceil \leq \lceil\log_{2}\lceil\frac{N}{2}\rceil\rceil\\
                             &\implies k-1 \leq \lceil\log_{2}(\lfloor\frac{N}{2}\rfloor+1)\rceil
\end{align*}

\noindent{\it Part 3:} Prove $f_{n,N}$ for $n \geq \lceil\log_{2}N\rceil$.

\noindent Base Cases: 
\begin{align*}
f(n, 1) &= 0 = (1-1)(n-\lceil\log_{2}1\rceil) + 2^{\lceil\log_{2}1\rceil}-1.
\end{align*}
\begin{align*}
f(\lceil\log_{2}N\rceil, N) = & 2^{\lceil\log_{2}N\rceil}-1 \text{ (from Part 1)}\\
                            = & (N-1)(\lceil\log_{2}N\rceil - \lceil\log_{2}N\rceil)\\
                              & + 2^{\lceil\log_{2}N\rceil}-1
\end{align*}

\noindent Inductive Step: Assume $f_{k-1, N'} = (N-1)(k-1-\lceil\log_{2}N'\rceil) + 2^{\lceil\log_{2}N'\rceil}-1$ $\forall N' \leq N$ and $k-1 \geq \lceil\log_{2}N'\rceil$. Then,
\begin{align*}
f_{k, N} &= f_{k-1, \lceil\frac{N}{2}\rceil} + f_{k-1, N-\lceil\frac{N}{2}\rceil+1} + 1\\
         &= f_{k-1, \lceil\frac{N}{2}\rceil} + f_{k-1, \lfloor\frac{N}{2}\rfloor+1} + 1\\
         &= (k-1-\lceil\log_{2}\lceil\frac{N}{2}\rceil\rceil)(\lceil\frac{N}{2}\rceil-1)+ 2^{\lceil\log_{2}\lceil\frac{N}{2}\rceil\rceil} -1\\
         &\ \ + (k-1-\lceil\log_{2}(\lfloor\frac{N}{2}\rfloor+1)\rceil)(\lfloor\frac{N}{2}\rfloor+1-1)\\
         &\ \ + 2^{\lceil\log_{2}(\lfloor\frac{N}{2}\rfloor+1)\rceil}-1+1\\
         &= (k-1-\lceil\log_{2}\frac{N}{2}\rceil)(\lceil\frac{N}{2}\rceil-1)+ 2^{\lceil\log_{2}\frac{N}{2}\rceil}-1\\
         &\ \ + (k-1-\lceil\log_{2}(\lfloor\frac{N}{2}\rfloor+1)\rceil)(\lfloor\frac{N}{2}\rfloor)+ 2^{\lceil\log_{2}(\lfloor\frac{N}{2}\rfloor+1)\rceil}
\end{align*}
\noindent For $N$ not a multiple of 2, $\lceil\log_{2}(\lfloor\frac{N}{2}\rfloor+1)\rceil = \lceil\log_{2}\frac{N}{2}\rceil$. In this case,
\begin{align*}
f_{k, N} &= (k-1-\lceil\log_{2}\frac{N}{2}\rceil)(\lceil\frac{N}{2}\rceil-1)+ 2^{\lceil\log_{2}\frac{N}{2}\rceil}\\
         &\ \ + (k-1-\lceil\log_{2}\frac{N}{2}\rceil)(\lfloor\frac{N}{2}\rfloor)+ 2^{\lceil\log_{2}\frac{N}{2}\rceil}-1\\
         &= (k-1-(\lceil(\log_{2}N)-1\rceil)(\lceil\frac{N}{2}\rceil + \lfloor\frac{N}{2}\rfloor - 1)\\
         &\ \ + 2^{\lceil\log_{2}\frac{N}{2}\rceil+1}-1\\
         &= (k-\lceil\log_{2}N\rceil)(N-1) + 2^{\lceil\log_{2}N\rceil}-1
\end{align*}
\noindent For $N$ a multiple of 2, $\lceil\log_{2}(\lfloor\frac{N}{2}\rfloor+1)\rceil = \lceil\log_{2}\frac{N}{2}\rceil+1$. In this case,
\begin{align*}
f_{k, N} &= (k-1-\lceil\log_{2}\frac{N}{2}\rceil)(\lceil\frac{N}{2}\rceil-1)+ 2^{\lceil\log_{2}\frac{N}{2}\rceil}-1\\
         &\ \ + (k-1-(\lceil\log_{2}\frac{N}{2}\rceil+1))(\lfloor\frac{N}{2}\rfloor)+ 2^{\lceil\log_{2}\frac{N}{2}\rceil+1}\\
         &= (k-1-\lceil\log_{2}\frac{N}{2}\rceil)(N-1) - \lfloor\frac{N}{2}\rfloor\\
         &\ \ + 2^{\lceil\log_{2}\frac{N}{2}\rceil}+2^{\lceil\log_{2}\frac{N}{2}\rceil+1}-1\\
         &= (k-1-(\lceil\log_{2}N\rceil-1))(N-1) - \frac{N}{2} + \frac{N}{2}\\
         &\ \ + 2^{\lceil(\log_{2}\frac{N}{2})+1\rceil}-1\\
         &= (k-\lceil\log_{2}N\rceil)(N-1) + 2^{\lceil\log_{2}N\rceil}-1
\end{align*}
\end{proof}

\section{Realizing Rational Distributions}

Given the previous results, a natural question arises: Can we also generate probability distributions over non-binary fractions? More generally, can we generate distributions for any rational distribution? This question was studied for the 2-state case by Wilhelm, Zhou, and Bruck\cite{WilhelmBruck}\cite{ZhouBruck} for distributions of the form $(\frac{a}{q^n}, \frac{b}{q^n})$. In their work, they demonstrated algorithms for realizing these distributions for any $q$ that is a multiple of 2 or 3. In addition, they proved that for any prime $q$, no algorithm exists that can generate all $(\frac{a}{q^n}, \frac{b}{q^n})$ using the switch set {\bf S} $ = \{\frac{1}{q}, \frac{2}{q}, ..., \frac{q-1}{q}\}$.

We approach the question of realizing rational $N$ state distributions from two angles and demonstrate algorithms for each. The key for these algorithms lie in a generalization of the Block-Interval construction of binary distributions.

\begin{lemma}[Block-Interval Construction of Distributions]
Let $p = (p_0, p_1, ..., p_{N-1})$ be any distribution. Then given any `cut', represented by the distribution $(q, 0, ..., 0, 1-q)$, and the index of the cut $k$ satisfying $q \geq \sum_{i=0}^{k-1}p_{i}$, $q \leq \sum_{i=0}^{k}p_{i}$,
\begin{align*}
&(p_0, p_1, ..., p_{N-1})\\
&=(\frac{p_0}{q}, ..., \frac{p_{k-1}}{q}, \frac{q - \sum_{i=0}^{k-1}p_{i}}{q}, 0, ...)\\
&\ \ \ + [(q,..., 1-q)(..., 0, \frac{\sum_{i=0}^{k}p_{i}-q}{1-q}, \frac{p_{k+1}}{1-q}, ..., \frac{p_{N-1}}{1-q})]\\
&=[(\frac{p_0}{q}, ..., \frac{p_{k-1}}{q}, \frac{q - \sum_{i=0}^{k-1}p_{i}}{q}, 0, ...) + (q, ..., 1-q)]\\
&\ \ \ *(..., 0, \frac{\sum_{i=0}^{k}p_{i}-q}{1-q}, \frac{p_{k+1}}{1-q}, ..., \frac{p_{N-1}}{1-q})
\end{align*}
\end{lemma}
\begin{proof}
The equation follows directly from the max-min composition rules.
\end{proof}
The name of the lemma comes from an intuitive representation of this equation. We can let probability distributions be represented by an interval covered with blocks. Each block corresponds to a different state of the distribution (a block can only be on the rightside of a block whose state it is greater than) and the ratio of the block length to the interval length represents the probability of the state. Then any cut along the interval will produce two separate block-intervals corresponding to exactly the distributions calculated above. This means that for any equation of this form, the block-interval cut will hold. Conversely, any cut of a block-interval corresponds to an equation with the appropriate distributions substituted. (see Figure \ref{visualization_rat}).

\begin{figure}[!t]
\centering
\includegraphics[width=3.0in]{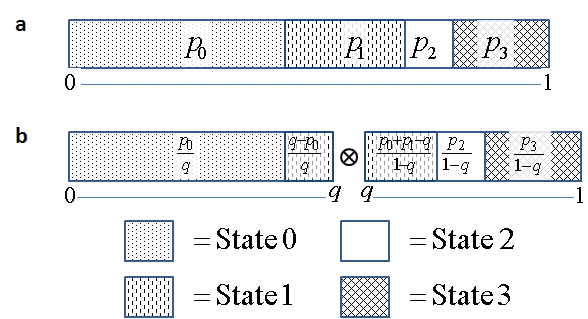}
\caption{\textbf{Block-Interval visualization (rational).} We want to realize $(p_0, p_1, p_2, p_3)$. A $\otimes$ symbol marks where the cut occurs, which corresponds to the use of one pswitch. \textit{(a)} We start with one interval split into blocks corresponding to the pswitch probabilities. \textit{(b)} When we cut with pswitch $(q, ..., 1-q)$, the interval is separated into 2 intervals: $[0, q]$ and $[q, 1]$.}
\label{visualization_rat}
\end{figure}

\subsection{State Reduction}

Given the results already proved on two state distributions, one natural way to tackle $N$ states is to first reduce the given $N$-state distribution into compositions of 2-state distributions. Then, if algorithms already exist for those 2-state forms, we are done.

\begin{theorem}[State Reduction Algorithm]
Using the switch set {\bf S} $ = \{\frac{1}{2}\}$, we can reduce any N-state distribution of the form $(\frac{x_0}{q}, \frac{x_1}{q}, ..., \frac{x_{N-1}}{q})$ into at most $N-1$ two-state distributions of the form $(..., \frac{x_i}{q}, ..., \frac{x_j}{q}, ...)$ using at most $f_{\log_{2}q, N} = (N-1)(\log_{2}q - \lceil\log_{2}N\rceil) + 2^{\lceil\log_{2}N\rceil}-1$ switches with the following algorithm - for an example, see Figure \ref{algthm}a:
\begin{algorithm}[H]
\begin{algorithmic}
\IF {$p \text{ is a 2 state distribution}$}
		\STATE return $p$
\ELSIF {$\frac{x_0}{q} > \frac{1}{2}$}
    \STATE return $\frac{1}{2}(\frac{2x_0 - q}{q}, \frac{2x_1}{q}, ..., \frac{2x_{N-1}}{q})$
\ELSIF {$\frac{x_0+x_1}{q} > \frac{1}{2}$}
    \STATE return $(\frac{2x_0}{q}, \frac{q - 2x_0}{q}, ...) + \frac{1}{2}(0, \frac{2x_0 + 2x_1 - q}{q}, \frac{2x_2}{q}, ...)$
\ELSIF {...}
    \STATE ...
\ELSIF {$\frac{x_0 + x_1 + ... + x_{N-1}}{q} > \frac{1}{2}$}
    \STATE return $(\frac{2x_0}{q}, \frac{2x_1}{q}, ..., \frac{2x_{N-1} - q}{q}) + \frac{1}{2}$
\ENDIF
\end{algorithmic}
\end{algorithm}
\end{theorem}

\begin{proof}
To prove this theorem, we need to show that 1) the breakdown of probabilities in each round is valid, 2) the algorithm will eventually terminate, and 3) the construction will have used $f_{\log_{2}q, N}$ pswitches, where $f_{n, N}$ is the complexity defined for binary $N$-state distributions. 

Notice that choosing $q = 2$ gives us the exact same algorithm as the binary $N$-state algorithm (except for the terminate clause). As with the binary case, Part 1 of the proof follows directly from the max-min composition rules.

We will now argue that the algorithm terminates. Let $k$ be the number of `rounds' the algorithm takes. Using the block-interval representation, we see that after each round, any interval will either have $\leq 2$ blocks (states) or the interval size will be $\frac{1}{2^k}$. Then when the interval size is less than $\frac{1}{q}$, each interval cannot have more than 3 blocks; otherwise, the middle block will be less than $\frac{1}{q}$, which is not possible except for a zero probability. Therefore, we have proved that the algorithm will eventually terminate.

To prove the complexity, we need to calculate the number of `rounds' before termination of the algorithm. This will happen at,
\begin{align*}
I &= \frac{1}{2^k} \leq \frac{1}{q}\\
  &\implies 2^k \geq q\\
  &\implies k \geq \log_{2}q
\end{align*}
Since the algorithm is identical to the binary $N$-state distribution, we can use the same complexity bound. In this case, $n$ corresponds to the $k$-th round when the algorithm terminates. Therefore, we use at most $(N-1)(\log_{2}q - \lceil\log_{2}N\rceil) + 2^{\lceil\log_{2}N\rceil}-1$ pswitches.
\end{proof}

\begin{figure}[!t]
\centering
\includegraphics[width=3.0in]{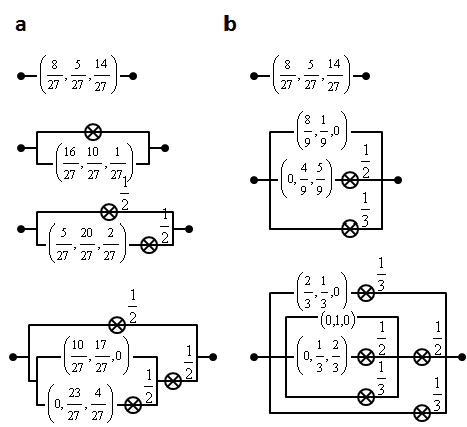}
\caption{\textbf{Algorithm Illustrations.} \textit{(a)} An example of the state recursion algorithm. \textit{(b)} An example of the denominator recursion algorithm.}
\label{algthm}
\end{figure}

\subsection{Denominator Reduction}

The second approach to realizing rational $N$-state distributions is to use a different switch set to directly reduces the power in the denominator. The intuition comes from a generalization of the block-interval construction. Rather than just cutting it at one point, i.e. at $\frac{1}{2}$ for realizing binary distributions, we can cut it at the points $\frac{1}{q}, \frac{2}{q}, ..., \text{ and }\frac{q-1}{q}$ to get $q$ equally sized intervals.

\begin{theorem}[Denominator Reduction Algorithm]
Using {\bf S} $= \{\frac{1}{2}, \frac{1}{3}, \frac{1}{4}, ..., \frac{1}{q}\}$ and the deterministic switches, we can realize any distribution $(\frac{x_0}{q^n}, \frac{x_1}{q^n}, ..., \frac{x_{N-1}}{q^n})$ using at most $(N-1)(q-1)(n-\lceil\log_{q}N\rceil) + q^{\lceil\log_{q}N\rceil}-1$ switches with the following algorithm - for an example, see Figure \ref{algthm}b:
\begin{algorithm}[H]
\begin{algorithmic}
\STATE Define the following quantities for $j = 0, 1, 2, ..., q$.
\STATE $k_j = $ the smallest index at which $\sum_{i=0}^{k_j}p_i > \frac{j}{q}$
\STATE $L_j = \frac{\sum_{i=0}^{k_{j-1}}x_i-(q-1)q^{n-1}}{q^{n-1}}$
\STATE $R_j = \frac{q^{n-1} - \sum_{i=0}^{k_{j}-1}x_i}{q^{n-1}}$
\IF {$p \in S$}
		\STATE return $p$
\ELSE 
    \STATE return $(\frac{x_0}{q^{n-1}}, ..., \frac{x_{k_1 - 1}}{q^{n-1}}, R_0, 0, ...)$
    \STATE  $ + \frac{1}{2}(..., 0, L_2, \frac{x_{k_1+1}}{q^{n-1}}, ..., \frac{x_{k_2-1}}{q^{n-1}}, R_2, 0, ...)$ 
    \STATE  $ + ...$
    \STATE  $ + \frac{1}{q-1}(..., 0, L_{q-1}, \frac{x_{k_{q-2}+1}}{q^{n-1}}, ..., \frac{x_{k_{q-1}-1}}{q^{n-1}}, R_{q-1}, 0, ...)$
    \STATE  $ + \frac{1}{q}(..., 0, L_q, \frac{x_{k_{q-1}+1}}{q^{n-1}}, ..., \frac{x_{N-1}}{q^{n-1}})$
\ENDIF
\end{algorithmic}
\end{algorithm}
\end{theorem}

\begin{proof}
The idea of the construction is to perform $q-1$ iterations of the block-interval construction lemma for each reduction of a pswitch. WLOG, let the original interval have length 1. Then the first cut at $\frac{q-1}{q}$ gives intervals of length $\frac{q-1}{q}$ and $\frac{1}{q}$. The second cut at $\frac{q-2}{q-1}$ of the $\frac{q-1}{q}$ interval gives intervals of length $\frac{q-1}{q}\frac{q-2}{q-1} = \frac{q-2}{q}$ and $\frac{q-1}{q}\frac{1}{q-1} = \frac{1}{q}$. The third cut leaves an interval of length $\frac{q-1}{q}\frac{q-2}{q-1}\frac{q-3}{q-2} = \frac{q-3}{q}$ and another $\frac{1}{q}$ interval. As this continues, we get $q$ intervals of length $\frac{1}{q}$. Our expression is exactly that of the block-interval construction lemma applied to the above description. Then since at each `round' of these decompositions, we reduce the denominator by 1, we will eventually terminate when $n = 0$, which is a deterministic switch.

The proof for the complexity is similar to the $N$-state binary proof. If we let $f_{q, n, N}$ be the maximum number of pswitches needed, then we have the following recursive relations:
\begin{align*}
f_{q, n, 1} &= 0\\
f_{q, 0, N} &= 0\\
f_{q, n, N} &= \max_{i_1, i_2, ..., i_{q}: \sum_{j}i_j = N + q-1}(q - 1 + \sum_{j}f_{q, n-1, i_j})
\end{align*}
One observation to note is that for $q = 2$, we get the same recursive functions as in the binary algorithm, which is expected:
\begin{align*}
f_{2, n, 1} &= 0\\
f_{2, 0, N} &= 0\\
f_{2, n, N} &= \max_{i_1, i_2: i_1 + i_2 = N-1}f_{2, n-1, i_1} + f_{2, n-1, i_2} + 1
\end{align*}

To go from this recursive form to the closed form version, we use similar methods to the binary $N$-state complexity function, so we will only outline the steps:
\begin{enumerate}
\item First prove that the recursive expression is maximized when all $i_j \in \{\lceil\frac{N}{2}\rceil, \lfloor\frac{N}{2}\rfloor + 1\}$. This is done by showing that the closed form for $f_{q, n, N}$ has a nonincreasing discrete derivative in $N$, i.e. $f_{q, n, N+1} - f_{q, n, N}$ is nonincreasing as $N$ increases.
\item Next show that the recursive expression evaluated at the appropriate $i_j$ will achieve the closed form by breaking it up into cases to simplify the ceiling and floor expressions.
\end{enumerate}
\end{proof}

\begin{corollary}
Let $q = p_1^{k_1} p_2^{k_2} ... p_{max}^{k_{max}}$, where $p_i$ are primes. Then using {\bf S} $= \{\frac{1}{2}, \frac{1}{3}, \frac{1}{4}, ..., \frac{1}{p_{max}}\}$, we can realize any distribution $(\frac{x_0}{q}, \frac{x_1}{q}, ..., \frac{x_{N-1}}{q})$ using at most $\sum_{i}f_{p_i, k_i, N}$ pswitches.
\end{corollary}
\begin{proof}
Here we are simply applying denominator reduction on each $p_i^{k_i}$. The complexity is just a summation of each step.
\end{proof}

\begin{corollary}
Let $q = p_1^{k_1} p_2^{k_2} ... p_{max}^{k_{max}}$, where $p_i$ are primes. Then using {\bf S} $= \{\frac{1}{2}, \frac{1}{3}, \frac{1}{5}, ..., \frac{1}{p_{max}}\}$, we can realize any distribution $(\frac{x_0}{q}, \frac{x_1}{q}, ..., \frac{x_{N-1}}{q})$.
\end{corollary}
\begin{proof}
Our claim here is that we can reduce the switch set to the inverse of the primes up to $p_{max}$. This is because using all of these switches, we can generate any other $\frac{1}{k}$, where $k$ is composite. However, this means that the number of pswitches will go up, i.e. $\frac{1}{4}$ will now use 2 switches $\frac{1}{2}*\frac{1}{2}$. We have not analyzed the complexity that this reduction in the switch set would result in.
\end{proof}

\section{Robustness of Probability Generation}

The above algorithms looked at probability generation given a fixed switch set of distributions. However, in physical systems, it may be the case that the generation of randomness is error-prone. If we want to use a pswitch with distribution $(p_0, p_1)$, the physical pswitch may actually have distribution $(p_0 + \epsilon, p_1 - \epsilon)$. Define the error of such a pswitch as $|\epsilon|$\cite{LohZhouBruck}. Loh, Zhou, and Bruck looked at generating 2-state probabilities given pswitches with errors. They found that any binary distribution generated according to their algorithm, regardless of size, had error bounded by $2\epsilon$. They also showed similar bounds for different distributions.

We examine the same problem in the context of multivalued distribution generation and show that a generalized result holds for any number of states. Define the error of a multivalued distribution as the largest error over all the states. That is, if a pswitch with desired distribution $(p_0, p_1, ..., p_{N-1})$ has an actual distribution of $(p_0 + \epsilon_0, p_1 + \epsilon_1, ..., p_{N-1} + \epsilon_{N-1})$, then the error of the switch is $\max_{i}|\epsilon_i|$. 

We will begin by demonstrating robustness for $N$-state binary distributions generated according to the algorithm in section IV. This algorithm uses switches from the switch set {\bf S} $ = \{\frac{1}{2}\}$ as well as the deterministic switches. For our analysis, we allow errors on the active states of pswitches. As a result, deterministic switches have no errors since the sum of the single active probability must equal 1. The $\frac{1}{2}$-pswitch has distribution $(\frac{1}{2} + \hat{\epsilon}, 0, ..., 0, \frac{1}{2} - \hat{\epsilon})$, $|\hat{\epsilon}| \leq \epsilon$.

\begin{lemma}[Error Bound on Boundary States]
Generate any distribution $(..., 0, \frac{x_i}{2^n}, ..., \frac{x_j}{2^n}, 0, ...)$ according to the binary $N$-state algorithm where state $i$ is the smallest active state and state $j$ is the largest active state. If we allow at most $\epsilon$ error on the pswitches in the switch set, then the actual distribution generated will be $(..., 0, \frac{x_i}{2^n} + \delta_i, ..., \frac{x_j}{2^n} + \delta_j, 0, ...)$ where $|\delta_i| \leq 2\epsilon$, $|\delta_j| \leq 2\epsilon$.
\end{lemma}
\begin{proof}
In each step of the algorithm, a distribution $r = (..., r_i, ..., r_j, ...)$ is made out of 2 pswitches $p = (..., p_i, ..., p_k, ...)$ and $q = (..., q_k, ..., q_j, ...)$ where $i \leq j \leq k$. The composition is in the form $r = p + (\frac{1}{2}, ..., \frac{1}{2})q$. We will prove robustness using induction on switches $p$ and $q$.

Base Case: We use the deterministic switches as our base case switches (for these, $i = k$ or $k = j$). These trivially satisfy the induction hypothesis since they have no error.

Inductive Step: Assume we are given $p$ and $q$ satisfying the inductive hypothesis. That is, for $p = (..., p_i + \Delta_i, ..., p_k + \Delta_k, ...)$, $q = (..., q_k + \delta_k, ..., q_j + \delta_j, ...)$, $|\Delta_i|, |\Delta_k|, |\delta_k|, |\delta_j| \leq 2\epsilon$. Then the errors for states $i$ and $j$ on distribution $r = p + (\frac{1}{2} + \hat{\epsilon}, ..., \frac{1}{2} - \hat{\epsilon})q$ can be calculated as follows:
\begin{align*}
|\text{error}_{i < k}| &= |(\frac{1}{2}+\hat{\epsilon})(p_i+\Delta_i)-\frac{1}{2}p_i|\\
                       &= |\frac{1}{2}\Delta_i + \hat{\epsilon}(p_i + \Delta_i)|\\
                       &\leq \frac{1}{2}|\Delta_i| + |\hat{\epsilon}||p_i + \Delta_i|\\
                       &\leq 2\epsilon\\
|\text{error}_{i = k}| &= |(\frac{1}{2}+\hat{\epsilon}) + (\frac{1}{2}-\hat{\epsilon})(q_i + \delta_i) - (\frac{1}{2} + \frac{1}{2}q_i)|\\
                       &= |\frac{1}{2}\delta_i + \hat{\epsilon}(1-(q_i + \delta_i))|\\
                       &\leq \frac{1}{2}|\delta_i| + |\hat{\epsilon}||1-(q_i + \delta_i)|\\
                       &\leq 2\epsilon\\
|\text{error}_{j > k}| &= |(\frac{1}{2}-\hat{\epsilon})(q_j+\delta_j)-\frac{1}{2}q_j|\\
                       &= |\frac{1}{2}\delta_j + \hat{\epsilon}(q_j + \delta_j)|\\
                       &\leq \frac{1}{2}|\delta_j| + |\hat{\epsilon}||q_j + \delta_j|\\
                       &\leq 2\epsilon\\
|\text{error}_{j = k}| &= |(\frac{1}{2}-\hat{\epsilon}) + (\frac{1}{2}+\hat{\epsilon})(p_j + \Delta_j) - (\frac{1}{2} + \frac{1}{2}p_j)|\\
                       &= |\frac{1}{2}\Delta_j + \hat{\epsilon}(1-(p_j + \Delta_j))|\\
                       &\leq \frac{1}{2}|\Delta_j| + |\hat{\epsilon}||1-(p_j + \Delta_j)|\\
                       &\leq 2\epsilon
\end{align*}
Then we find that the errors for states $i$ and $j$ still satisfy the inductive hypothesis, so we are done.
\end{proof}

\begin{theorem}[Robustness of Binary N-state Distributions]
Generate any distribution $(\frac{x_0}{2^n}, \frac{x_1}{2^n}, ..., \frac{x_{N-1}}{2^n})$ according to the binary $N$-state algorithm. If we allow the pswitches in the switch set {\bf S} $= \frac{1}{2}$ to have up to $\epsilon$ error in the active states, then the actual distribution $(\frac{x_0}{2^n} + \delta_0, \frac{x_1}{2^n} + \delta_1, ..., \frac{x_{N-1}}{2^n} + \delta_{N-1})$ has errors $|\delta_i| \leq 3\epsilon$, $|\delta_0|, |\delta_{N-1}| \leq 2\epsilon$.
\end{theorem}
\begin{proof}
In each step of the binary $N$-state algorithm, a distribution $r = (..., r_i, ..., r_j, ...)$ is made out of 2 pswitches $p = (..., p_i, ..., p_k, ...)$ and $q = (..., q_k, ..., q_j, ...)$ where $i \leq j \leq k$, $i$ is the smallest active state and $j$ is the largest active state. The composition is in the form $r = p + (\frac{1}{2}, ..., \frac{1}{2})q$. We will prove robustness using induction on switches $p$ and $q$.

Base Case: We use the deterministic switches as our base case switches. These trivially satisfy the induction hypothesis since they have no error.

Inductive Step: Assume we are given $p$ and $q$ satisfying the inductive hypothesis. That is, for $p = (..., p_i + \Delta_i, ..., p_k + \Delta_k, ...)$, $q = (..., q_k + \delta_k, ..., q_j + \delta_j, ...)$, $|\Delta_s|, |\delta_s| \leq 3\epsilon$, $|\Delta_i|, |\Delta_k|, |\delta_k|, |\delta_j| \leq 2\epsilon$. Then the errors for states $s$ on distribution $r = p + (\frac{1}{2} + \hat{\epsilon}, ..., \frac{1}{2} - \hat{\epsilon})q$ can be calculated as follows:
\begin{enumerate}
\item For $s = i, j$, we apply the previous lemma on boundary states
\item For $i < s < k, k < s < j$, we have a similar expression as calculated in the proof for boundary error (steps skipped below since shown earlier).
\begin{align*}
|\text{error}_{i < s < k}| &\leq \frac{1}{2}|\Delta_s| + |\hat{\epsilon}||p_s + \Delta_s|\\
                       &\leq \frac{5}{2}\epsilon \leq 3\epsilon\\
|\text{error}_{k < s < j}| &\leq \frac{1}{2}|\delta_s| + |\hat{\epsilon}||q_s + \delta_s|\\
                       &\leq \frac{5}{2}\epsilon \leq 3\epsilon
\end{align*}
\item For $s = k$, we have the following error:
\begin{align*}
|\text{error}_{s = k}| &= |(\frac{1}{2} + \hat{\epsilon})(p_k + \Delta_k) + (\frac{1}{2} - \hat{\epsilon})(q_k + \delta_k)\\
                       &\ \ \  - \frac{1}{2}(p_k + q_k)|\\
                       &= |\frac{1}{2}\Delta_k + \frac{1}{2}\delta_k + \hat{\epsilon}[(p_k + \Delta_k) - (q_k + \delta_k)]|\\
                       &\leq \frac{1}{2}(|\Delta_k| + |\delta_k|) + \epsilon\\
                       &\leq 3\epsilon
\end{align*}
\end{enumerate}
Then we find that the errors for all states $s$ still satisfy the inductive hypothesis, so for $i = 0$ and $j = N-1$, we are done.
\end{proof}

Now we will show that the $N$-state rational distributions generated according to the algorithms in Section V are also robust. For the state reduction algorithm, all the results are identical to the binary robustness results since the proofs did not assume that $p$ and $q$ were binary distributions. However, since the state reduction algorithm depends on other unknown 2-state algorithms, the error will depend on the errors generated from the 2-state algorithm. The robustness result for the denominator reduction algorithm is very similar to the one for binary distributions. The main difference is that we have a different switch set {\bf S} $= \{\frac{1}{2}, \frac{1}{3}, ..., \frac{1}{q}\}$. 

\begin{theorem}[Robustness of Rational Distributions]
Generate any distribution $(\frac{x_0}{q^n}, \frac{x_1}{q^n}, ..., \frac{x_{N-1}}{q^n})$ according to the denominator reduction algorithm. If we allow the pswitches in the switch set {\bf S} $= \{\frac{1}{2}, ..., \frac{1}{q}\}$ to have up to $\epsilon$ error in the active states, then the actual distribution $(\frac{x_0}{q^n} + \delta_0, \frac{x_1}{q^n} + \delta_1, ..., \frac{x_{N-1}}{q^n} + \delta_{N-1})$ has errors $|\delta_i| \leq (q+1)\epsilon$, $|\delta_0|, |\delta_{N-1}| \leq q\epsilon$.
\end{theorem}
\begin{proof}
For generating rational distributions, each iteration repeats compositions of the form:
\begin{align*}
(z_0, ..., z_{N-1}) &= (p_0, ..., p_{i_1}, 0, ...)\\
                    & + \frac{1}{2}(..., 0, q_{i_1}, ..., q_{i_2}, 0, ...)\\
                    & + \frac{1}{3}(..., 0, r_{i_2}, ..., r_{i_3}, 0, ...)\\
                    & + ...\\
                    & + \frac{1}{q}(..., 0, y_{i_{q-1}}, ..., y_{i_q}, 0, ...)
\end{align*}
We notice that each composition of this form is the same as $q-1$ compositions of the form $(..., p_i, ..., p_k, ...) + \frac{1}{\hat{q}}(..., q_k, ..., q_j, ...)$ where $\hat{q} = 2, 3, ... q$. Then from here, the work is almost identical to the proof of the binary distributions, so we will only outline the steps.
\begin{enumerate}
\item Prove that the boundary states of distributions formed from $(..., p_i, ..., p_k, ...) + \frac{1}{\hat{q}}(..., q_k, ..., q_j, ...)$ compositions, $\hat{q} \in \{2, 3, ..., q\}$ have error of at most $q\epsilon$ regardless of the number of compositions (use induction). As before, we allow pswitches $(\frac{\hat{q}-1}{\hat{q}} + \hat{\epsilon}, ..., \frac{1}{\hat{q}} - \hat{\epsilon})$ to have error $\hat{\epsilon} \leq \epsilon$.
\item Using the boundary state error, prove that the errors of all states have error of at most $(q-1)\epsilon$.
\end{enumerate}
Then, since all the compositions making up the rational distributions are of the form $(..., p_i, ..., p_k, ...) + \frac{1}{\hat{q}}(..., q_k, ..., q_j, ...)$, the steps above are sufficient for proving our theorem.
\end{proof}

\section{Universal Probability Generator}

Up till now, we have only looked at circuits with set switches and described algorithms for realizing specific probability distributions. A next question is: What about circuits that implement stochastic functions? That is, how can we reason about circuits that implement different probability distributions given `input' switches? Wilhelm and Bruck approached this problem for the 2-state version\cite{WilhelmBruck} by constructing a circuit which they called a Universal Probability Generator (UPG). The function of this circuit was to map $n$-deterministic bits into output probabilities of the form $\frac{x_0}{2^n}$ in increasing order. 

This functionality alone is not surprising since it can easily be done with an exponential number of switches using a tree-like structure; however, the remarkable result is that the UPG only requires a linear number of switches in $n$. In this section we will show a generalization of this circuit which is able to realize any binary probability distribution of the form $(\frac{x_0}{2^n}, \frac{x_1}{2^n}, ..., \frac{x_{N-1}}{2^n})$ using a number of switches that is polynomial in $n$ with degree $N-1$.

The strategy we will take for deriving the UPG construction is first to ask: `How should the UPG function?' Then we will build a circuit that naively implements the functionality and use algebraic rules to reduce the complexity of the circuit.

\subsection{2-state UPG}

We will first review the 2-state UPG of Wilhelm and Bruck\cite{WilhelmBruck}, deriving the results in a manner that is similar to the steps we take to get the generalized $N$-state UPG. 

\begin{figure}[!b]
\centering
\includegraphics[width=3.0in]{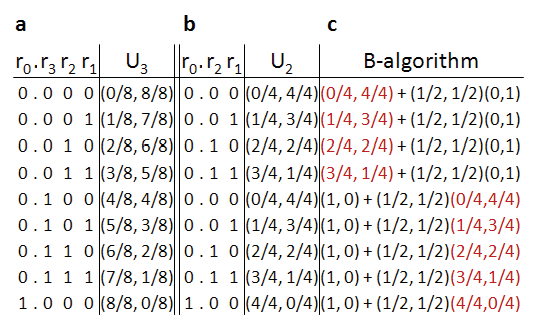}
\caption{\textbf{2-state Universal Probability Generator (UPG) mappings} \textit{(a)} A mapping for a UPG that generates distributions of the form $(\frac{x_0}{8}, \frac{x_1}{8})$. \textit{(b)} Removing $r_3$, we look at the inputs to a UPG that generates distributions of the form $(\frac{x_0}{4}, \frac{x_1}{4})$. \textit{(c)} Notice that the B-algorithm decompositions of the probabilities of part (a) correspond to those in part (b).}
\label{UPG2TAB}
\end{figure}

{\it Definition: (2-state UPG)}.
A 2-state UPG $U_{n}$ is a circuit that realizes distributions of the form $(\frac{x_0}{2^n}, \frac{x_1}{2^n})$ using $n+1$ input bits which we will refer to as $r_0$, $r_1$, ..., $r_n$. When the input bits, in the order $r_0, r_n, r_{n-1},$ ..., $r_2, r_1$ are set to the binary representation of $\frac{x_0}{2^n}$, then the circuit $U_{n}$ will realize distribution $(\frac{x_0}{2^n}, \frac{x_1}{2^n})$. In other words, to realize any desired binary probability $(\frac{x_0}{2^n}, \frac{x_1}{2^n})$, we set the input bits to the state 0 probability. 

As an example, we look at the input-output mappings for the circuit $U_{3}$. If we input $\mathbf{r} = 0001$, the circuit will realize $(\frac{1}{8}, \frac{7}{8})$ since $\frac{1}{8} = 0.001_2$. The input $0101$ will realize $(\frac{5}{8}, \frac{3}{8})$ since $\frac{5}{8} = 0.101_2$ (see Figure \ref{UPG2TAB}a).

The motivation for the UPG circuit comes from an interesting property in the truth table. For example, let us enumerate all the outputs for $U_{3}$ (Fig. \ref{UPG2TAB}a). For each row (input), we ask the following questions: What would the output of $U_{2}$ be given the inputs $r_0, r_2, r_1$ that were used for $U_{3}$? Is there a relationship to the construction of the $U_{3}$ output probability?

If we calculate these outputs, we find that they are the same probability distributions used in the binary algorithm for 2-states (Fig. \ref{UPG2TAB}bc). From here, the (exponential) recursive construction is straightforward.

\begin{figure}[!t]
\centering
\includegraphics[width=3.0in]{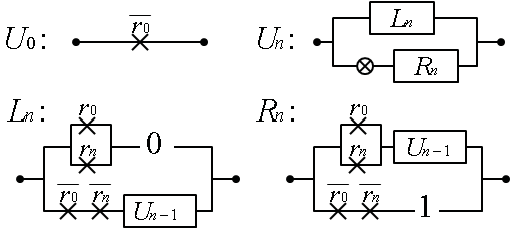}
\caption{\textbf{2-state exponential UPG} This is a UPG derived directly from the B-algorithm. It uses an exponential number of switches since $U_{n}$ uses two copies of $U_{n-1}$ in its recursive construction.}
\label{UPG2EXP}
\end{figure}

\begin{lemma}
A 2-state UPG $U_{n}$ with inputs $r_0, r_n, \ldots, r_2, r_1$ can be constructed with an exponential number of switches using the recursive construction in Figure \ref{UPG2EXP}, where the $n$ bits used in $U_{n-1}$ are $r_0, r_{n-1}, \ldots, r_2, r_1$.
\end{lemma}
\begin{proof}
We can prove this inductively. 

Base Case: $U_{0} = \bar{r}_0$ realizes $(0, 1)$ and $(1, 0)$ when $r_0 = 0$ and $r_0 = 1$ respectively. 

Inductive Step: Let $r_{0n} = r_0 + r_n$, which also implies $\bar{r}_{0n} = \bar{r}_0\bar{r}_n$. Assume $U_{n-1}$ realizes the correct distributions given the defined inputs. Then we want to show that
\begin{align*}
U_{n} &=L_{n} + (\frac{1}{2}, \frac{1}{2})R_{n},\\
L_{n} &=r_{0n}\cdot 0 + \bar{r}_{0n}U_{n-1},\\
R_{n} &=r_{0n}U_{n-1} + \bar{r}_{0n}\cdot 1,
\end{align*}
This is equivalent to showing,
\begin{align*}
U_{n} &=L_{n} + (\frac{1}{2}, \frac{1}{2})R_{n},\\
L_{n} &=
\begin{cases}
(1, 0),  &\text{ if }r_{0n} = 1\\
U_{n-1}, &\text{ if }\bar{r}_{0n} = 1
\end{cases}\\
R_{n}  &=
\begin{cases}
U_{n-1}, &\text{ if }r_{0n} = 1\\
(0, 1),  &\text{ if }\bar{r}_{0n} = 1
\end{cases}
\end{align*}
which is equivalent to showing,
\begin{align*}
U_{n} &=
\begin{cases}
(1, 0) + (\frac{1}{2}, \frac{1}{2})U_{n-1}, &\text{ if }r_{0n} = 1\\
U_{n-1} + (\frac{1}{2}, \frac{1}{2})(0, 1), &\text{ if }\bar{r}_{0n} = 1
\end{cases}
\end{align*}
Here, we are reminded of the B-algorithm (for realizing any 2-state binary distribution):
\begin{align*}
(\frac{x_0}{2^n}, \frac{x_1}{2^n}) &=
\begin{cases} 
(1, 0) + (\frac{1}{2}, \frac{1}{2})(\frac{x_0}{2^{n-1}}-1, \frac{x_1}{2^{n-1}}), &\text{ if }\frac{x_0}{2^n} \geq \frac{1}{2}\\
(\frac{x_0}{2^{n-1}}, \frac{x_1}{2^{n-1}}-1) + (\frac{1}{2}, \frac{1}{2})(0, 1), &\text{ if }\frac{x_0}{2^n} < \frac{1}{2}
\end{cases}
\end{align*}
Then we are almost done with our proof. We know that $r_{0n} = 1 \iff \frac{x_0}{2^n} \geq \frac{1}{2}$ and $\bar{r}_{0n} = 1 \iff \frac{x_0}{2^n} < \frac{1}{2}$. In addition, if $\mathbf{r}$ is set to generate $(\frac{x_0}{2^n}, \frac{x_1}{2^n})$, then using $r_0, r_{n-1}, ..., r_2, r_1$ as inputs to $U_{n-1}$, we get
\begin{align*}
U_{n-1} &=
\begin{cases}
(\frac{x_0}{2^{n-1}}-1, \frac{x_1}{2^{n-1}}), &\text{ if }r_{0n} = 1\\
(\frac{x_0}{2^{n-1}}, \frac{x_1}{2^{n-1}}-1), &\text{ if }\bar{r}_{0n} = 1
\end{cases}
\end{align*}
This is because when $r_0 = 0$, removing $r_n$ is like shifting the fractional bits left (multiplying by 2) and then subtracting 1 if the bit removed was 1. When $r_0 = 1$, then removing $r_n$ doesn't change anything, which still satisfies $1 = 1(2) - 1$. At this point, we can invoke the B-algorithm to conclude that $U_{n}$ will successfully realize $(\frac{x_0}{2^n}, \frac{x_1}{2^n})$. 
\end{proof}

\begin{figure}[!t]
\centering
\includegraphics[width=3.0in]{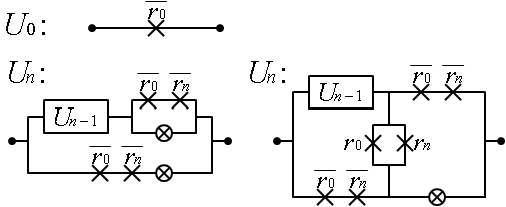}
\caption{\textbf{2-state UPG reduced} After reducing the exponential UPG, we get two linear UPGs. One is a sp-circuit and uses 2 stochastic switches. The other is non-sp and uses 1 stochastic switch.}
\label{UPG2RED}
\end{figure}

\begin{lemma}
A 2-state UPG $U_{n}$ with inputs $r_0, r_n, \ldots, r_2, r_1$ can be constructed with a linear number of switches using either of the two recursive constructions in Figure \ref{UPG2RED}, where the $n$ bits used in $U_{n-1}$ are $r_0, r_{n-1}, \ldots, r_2, r_1$.
\end{lemma} 
\begin{proof}
We will prove this by algebraically reducing the exponential UPG. Let $p_i = (\frac{1}{2}, \frac{1}{2})$ be i.i.d. 
\begin{align*}
U_{n} &= L_{n} + p_nR_{n}\\
         &= r_{0n}\cdot 0 + \bar{r}_{0n}U_{n-1} + p_n(r_{0n}U_{n-1} + \bar{r}_{0n}\cdot 1)\\
         &= \bar{r}_{0n}U_{n-1} + p_n(r_{0n}U_{n-1} + \bar{r}_{0n})\\
         &= U_{n-1}(\bar{r}_{0n} + r_{0n}p_n) + \bar{r}_{0n}p_n\\
         &= U_{n-1}(\bar{r}_{0n} + p_n) + \bar{r}_{0n}p_n
\end{align*}
This is exactly the form of the series parallel construction in Fig. \ref{UPG2RED}
\begin{align*}
U_{n} &= L_{n} + p_nR_{n}\\
         &= r_{0n}\cdot 0 + \bar{r}_{0n}U_{n-1} + p_n(r_{0n}U_{n-1} + \bar{r}_{0n}\cdot 1)\\
         &= \bar{r}_{0n}U_{n-1} + p_n(r_{0n}U_{n-1} + \bar{r}_{0n})\\
         &= [U_{n-1}][\bar{r}_{0n}] + [\bar{r}_{0n}][r_{0n}][\bar{r}_{0n}]\\
         &\ \ \ \ \ + [U_{n-1}][r_{0n}][p_n] + [\bar{r}_{0n}][p_n]
\end{align*}
This is exactly the form of the non-series parallel construction in Fig. \ref{UPG2RED}
\end{proof}
The result thus far is nice, but one feels somewhat unsatisfied at the number of times $r_0$ must be used. It is strange that $r_0$ is used so many times in the circuit construction even though it is only set to 1 for a single distribution in the truth table - the deterministic distribution $(1, 0)$. We solve this problem to get to the final form of Wilhelm and Bruck's 2-state UPG.

\begin{figure*}[!b]
\centering
\includegraphics[width=6.5in]{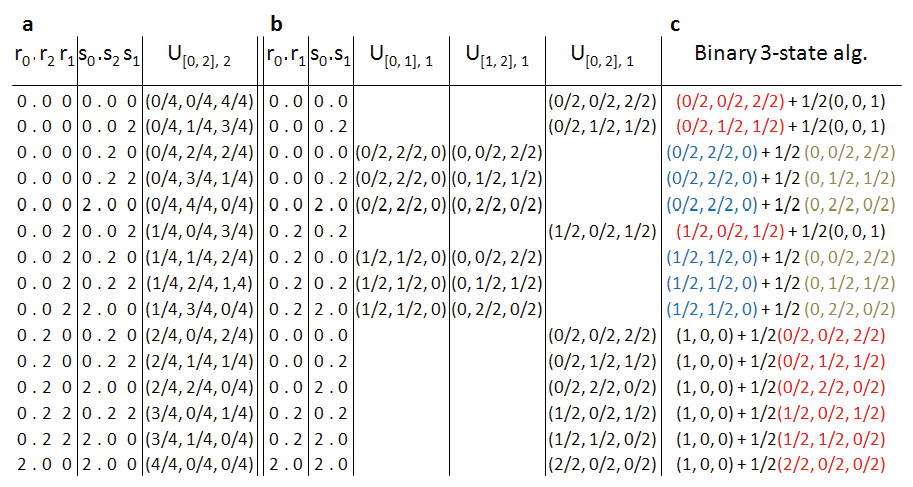}
\caption{\textbf{3-state UPG mapping} The mapping for a UPG that generates distributions of the form $(\frac{x_0}{4}, \frac{x_1}{4}, \frac{x_2}{4})$. Here we notice that $y_0$ is being used a couple of times. Again, we don't make predictions about inputs that don't correspond to a probability distribution: this includes both those with a sum greater than 1 and those for which the $y$ encoding is less than the $x$ encoding (which would mean a negative $\frac{x_1}{2^n}$ value).}
\label{UPG3TAB}
\end{figure*}

\begin{theorem}
A 2-state UPG $U_{n}$ with inputs $r_0, r_n, \ldots, r_2, r_1$ can be constructed with a linear number of switches using either of the two recursive constructions in Figure \ref{UPG2NOBIT}, where the n bits used in $U_{n-1}$ are $r_0, r_{n-1}, \ldots, r_2, r_1$.
\end{theorem}

\begin{proof}
We prove that the circuits in Figure \ref{UPG2RED} and Figure \ref{UPG2NOBIT} are equivalent by using induction.
Base Case: $U_{0} = \bar{r}_0 = \bar{r}_0U'_{0}$

Inductive Step:
Assume that $U_{n-1} = r_0U'_{n-1}$. Then,
\begin{align*}
U_{n} &= U_{n-1}(\bar{r}_n\bar{r}_0 + p_n) + p_n\bar{r}_n\bar{r}_0\\
         &= \bar{r}_0U'_{n-1}(\bar{r}_n\bar{r}_0 + p_n) + p_n\bar{r}_n\bar{r}_0\\
         &= \bar{r}_0(U'_{n-1}(\bar{r}_n\bar{r}_0 + p_n) + p_n\bar{r}_n)\\
         &= \bar{r}_0(U'_{n-1}(\bar{r}_n + p_n) + p_n\bar{r}_n)\\
         &= \bar{r}_0U'_{n}
\end{align*}
\end{proof}

In this final form, the series-parallel circuit uses $2n$ pswitches and $2n+1$ deterministic switches. The non-sp construction uses $n$ pswitches and $3n+1$ deterministic switches.

Before we can generalize to $N$ states, we need to define some new notation and equalities. Define the 2-state UPG $U_{[i, i+1], n}$ to be a circuit generating probability $\frac{x_0}{2^n}$ on state $i$ and $\frac{x_1}{2^n}$ on state $i+1$.

Under this notation, the previous results we proved were for $U_{[0, 1], n}$. We can extend these results to $U_{[i, i+1], n}$ with the following changes:
\begin{enumerate}
\item The input bits $\mathbf{r}$ will take on values $i$ and $i+1$ instead of $0$ and $1$ respectively
\item The pswitch $p_n$ will take on values $i$ and $i+1$ with $\frac{1}{2}$ probability each instead of taking on values $0$, $1$ with equal probability. 
\end{enumerate}
But what if we didn't use values $i$ and $i+1$ for $\mathbf{r}, p_n$? 
\begin{lemma}
Let $U_{[i, i+1], n}(a, b)$ be a circuit identical to $U_{[i, i+1], n}$ except that $\mathbf{r}, p_n$ take on values $a$ and $b$. Our previous results for probability generation are for $U_{[i, i+1], n}(i, i+1)$. Then if $a \leq i$, $b \geq i+1$,
\begin{align*}
U_{[i, i+1], n}(i, i+1) = (i + U_{[i, i+1], n}(a, b))(i+1)
\end{align*}
\end{lemma}
\begin{proof}
$U_{[i, i+1], n}(a, b)$ will realize the desired distribution on states $a$ and $b$. Then if $a \leq i$, $b \geq i+1$, it is trivially that taking the max of $i$ and the min of $i+1$ will give us the same distribution on states $i$ and $i+1$.
\end{proof}
We want to avoid this messy notation for future generalizations. All future instances of $U_{[i, i+1], n}$ are actually representing $(i + U_{[i, i+1], n})(i+1)$. In other words, $U_{[i, i+1], n}$ will always generate the `correct' distribution on states $i$ and $i+1$ as long as $\mathbf{r}, p_n$ take on values $a \leq i$ and $b \geq i+1$.

\begin{figure}[!t]
\centering
\includegraphics[width=3.0in]{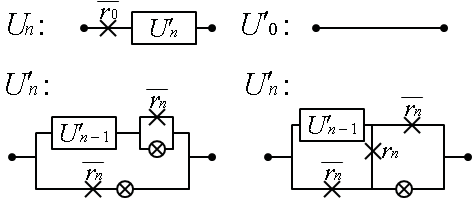}
\caption{\textbf{2-state UPG removed bit} Even though $r_0$ was important in controlling the cases of our algorithm, it turns out we can remove it from the recursive construction and just append it at the end. Note that $U'_{n}$ only uses bits $r_n, \ldots, r_2, r_1$}
\label{UPG2NOBIT}
\end{figure}

\subsection{3-state UPG}

Now we will derive the 3-state UPG by generalizing the steps we used in the 2-state UPG derivation. We first construct an exponential 3-state UPG that follows closely from the algorithm for realizing 3-state distributions. Then, we will algebraically reduce it to a quadratic construction and remove unnecessarily repetitive bits.

{\it Definition: (3-state UPG)}.
A 3-state UPG $U_{[0, 2], n}$ is a circuit that generates distributions of the form $(\frac{x_0}{2^n}, \frac{x_1}{2^n}, \frac{x_2}{2^n})$ using 2 sets of $n+1$ input bits $\mathbf{r}, \mathbf{s}$. These input bits will take on values $0$ or $2$; when the bits $\mathbf{r}$, in the order $r_0, r_n, r_{n-1}, \ldots, r_2, r_1$ are set to the binary representation of $\frac{x_0}{2^n}$ and the bits $\mathbf{s}$, in the order $s_0, s_n, s_{n-1}, \ldots, s_2, s_1$ are set to the binary representation of $\frac{x_0}{2^n} + \frac{x_1}{2^n}$ (with the symbol 2 replacing the boolean 1), then the circuit $U_{[0, 2], n}$ will realize distribution $(\frac{x_0}{2^n}, \frac{x_1}{2^n}, \frac{x_2}{2^n})$. In other words, to realize any desired binary probability $(\frac{x_0}{2^n}, \frac{x_1}{2^n}, \frac{x_2}{2^n})$, we set the input bits $\mathbf{r}$ and $\mathbf{s}$ to $p_0$ and $p_0 + p_1$ respectively.

As an example, we look at the input-output mappings for the circuit $U_{[0, 2], 2}$. If we input $\mathbf{r} = 002, \mathbf{s} = 020$, the circuit will realize $(\frac{1}{4}, \frac{1}{4}, \frac{1}{2})$ since $\frac{1}{4} = 0.01_2$ and $\frac{1}{4} + \frac{1}{4} = \frac{1}{2} = 0.10_2$. The input $\mathbf{r} = 020, \mathbf{s} = 022$ will realize $(\frac{1}{2}, \frac{1}{4}, \frac{1}{4})$ since $\frac{1}{2} = 0.10_2$ and $\frac{1}{2} + \frac{1}{4} = \frac{3}{4} = 0.11_2$. (see Figure \ref{UPG3TAB}a).

Again, the motivation for the exponential UPG comes from an interesting property in the truth table of $U_{[0, 2], 2}$. We first eumerate all outputs given inputs corresponding to valid probability distributions. For each row (input), we ask: What are the outputs of $U_{[0, 1], 1}, U_{[1, 2], 1}, U_{[0, 2], 1}$ if we use the inputs $r_0, r_1$ for $U_{[0, 1], 1}$, the inputs $s_0, s_1$ for $U_{[1, 2], 1}$, and $r_0, r_1, s_0, s_1$ for $U_{[0, 2], 2}$? Is there a relationship to the construction of the $U_{[0, 2], 1}$ output probability. 

The answer is yes. We find that they are the same probability distributions that are used in the binary algorithm for 3-states (Fig. \ref{UPG3TAB}bc). We can now derive the exponential recursive construction.

\begin{figure}[!t]
\centering
\includegraphics[width=3.0in]{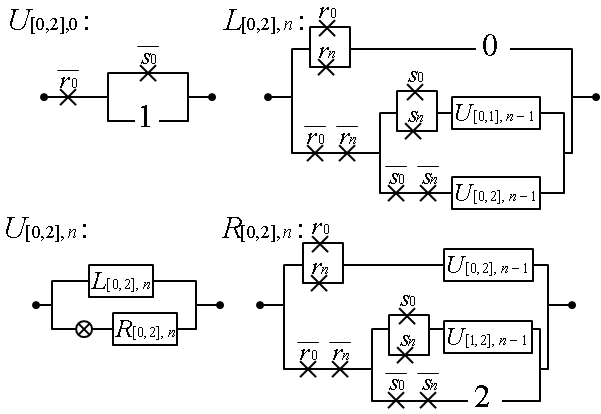}
\caption{\textbf{3-state exponential UPG} We design a UPG directly from the algorithm. Note that this uses an exponential number of switches since the recursion uses two copies of $U_{[0, 2], n-1}$.}
\label{UPG3EXP}
\end{figure}

\begin{lemma}
A 3-state UPG $U_{[0, 2], n}$ with inputs $r_0, r_n, \ldots, r_2, r_1$ and $s_0, s_n, \ldots, s_2, s_1$ can be constructed with an exponential number of switches using the recursive construction in Figure \ref{UPG3EXP}, where the bits used in $U_{[0, 1], n-1}$ are $r_0, r_{n-1}, \ldots, r_2, r_1$, the bits used in $U_{[1, 2], n-1}$ are $s_0, s_{n-1}, \ldots, s_2, s_1$, and the bits used in $U_{[0, 2], n-1}$ are $r_0, r_{n-1}, \ldots, r_2, r_1$ and $s_0, s_{n-1}, \ldots, s_2, s_1$.
\end{lemma}
\begin{proof}
We can prove this inductively.

Base Case: $U_{[0, 2], 0} = \bar{r}_0(\bar{s}_0+1)$ realizes $(1, 0, 0)$ when $r_0 = 2$, $(0, 1, 0)$ when $r_0 = 0, s_0 = 2$, and $(0, 0, 1)$ when $r_0 = 0, s_0 = 0$.

Inductive Step: Assume $U_{[0, 2], n-1}$ realizes the correct distributions given the defined inputs. Then we want to show that
\begin{align*}
U_{[0, 2], n} &=L_{[0, 2], n} + (\frac{1}{2}, 0, \frac{1}{2})R_{[0, 2], n},\\
L_{[0, 2], n} &=r_{0n}\cdot 0 + \bar{r}_{0n}(s_{0n}U_{[0, 1], n-1} + \bar{s}_{0n}U_{[0, 2], n-1}),\\
R_{[0, 2], n} &=r_{0n}U_{[0, 2], n-1} + \bar{r}_{0n}(s_{0n}U_{[1, 2], n-1} + \bar{s}_{0n}\cdot 2)
\end{align*}
This is equivalent to showing,
\begin{align*}
U_{r, n} &=L_{r,n} + (\frac{1}{2}, 0, \frac{1}{2})R_{r, n},\\
L_{r, n} &=
\begin{cases}
(1, 0, 0),       &\text{ if }r_{0n} = 2\\
U_{[0, 1], n-1}, &\text{ if }\bar{r}_{0n} = 2, s_{0n} = 2\\
U_{[0, 2], n-1}, &\text{ if }\bar{r}_{0n} = 2, \bar{s}_{0n} = 2
\end{cases}\\
R_{r,n}  &=
\begin{cases}
U_{[0, 2], n-1}, &\text{ if }r_{0n} = 2\\
U_{[1, 2], n-1}, &\text{ if }\bar{r}_{0n} = 2, s_{0n} = 2\\
(0, 0, 1),       &\text{ if }\bar{r}_{0n} = 2, \bar{s}_{0n} = 2
\end{cases}
\end{align*}
which is equivalent to showing,
\begin{align*}
U_{[0, 2], n} &=
\begin{cases}
(1, 0, 0) + (\frac{1}{2}, 0, \frac{1}{2})U_{[0, 2], n-1}, &\text{ if }r_{0n} = 2\\
U_{[0, 1], n-1} + (\frac{1}{2}, 0, \frac{1}{2})U_{[1, 2], n-1}, &\text{ if }\bar{r}_{0n} = 2, s_{0n} = 2\\
U_{[0, 2], n-1} + (\frac{1}{2}, 0, \frac{1}{2})(0, 0, 1), &\text{ if }\bar{r}_{0n} = 2, \bar{s}_{0n} = 2
\end{cases}
\end{align*}
Here, we are reminded of the algorithm for realizing any 3-state binary distribution: $(\frac{x_0}{2^n}, \frac{x_1}{2^n}, \frac{x_2}{2^n}) =$
\begin{align*}
\begin{cases} 
(1, 0, 0)                                                         + \frac{1}{2}(\frac{x_0}{2^{n-1}}-1, \frac{x_1}{2^{n-1}}, \frac{x_2}{2^{n-1}}), &\text{ if }\frac{x_0}{2^n} \geq \frac{1}{2}\\
(\frac{x_0}{2^{n-1}}, \frac{x_1+x_2}{2^{n-1}}-1, 0)               + \frac{1}{2}(0, \frac{x_0 + x_1}{2^{n-1}}-1, \frac{x_2}{2^{n-1}}),             &\text{ if }\frac{x_0}{2^n} < \frac{1}{2},\\
                                                                                                                                                  &\frac{x_0 + x_1}{2^n} \geq \frac{1}{2}\\
(\frac{x_0}{2^{n-1}}, \frac{x_1}{2^{n-1}}, \frac{x_2}{2^{n-1}}-1) + \frac{1}{2}(0, 0, 1),                                                         &\text{ if }\frac{x_0}{2^n} < \frac{1}{2},\\
                                                                                                                                                  &\frac{x_0 + x_1}{2^n} < \frac{1}{2}
\end{cases}
\end{align*}

From here, it is straightforward to complete our proof. $r_{0n} = 2 \iff \frac{x_0}{2^n} \geq \frac{1}{2}$, $\bar{r}_{0n} = 2 \iff \frac{x_0}{2^n} < \frac{1}{2}$, $s_{0n} = 2 \iff \frac{x_0 + x_1}{2^n} \geq \frac{1}{2}$, and $\bar{s}_{0n} = 2 \iff \frac{x_0 + x_1}{2^n} < \frac{1}{2}$. In addition, if $\mathbf{r}, \mathbf{s}$ are set to generate $(\frac{x_0}{2^n}, \frac{x_1}{2^n}, \frac{x_2}{2^n})$, then using $r_0, r_{n-1}, \ldots, r_2, r_1$ and $s_0, s_{n-1}, \ldots, s_2, s_1$ as inputs to $U_{[0, 1], n-1}$, $U_{[1, 2], n-1}$ and $U_{[0, 2], n-1}$ we get the corresponding values to the 3-state binary algorithm. Then we can conclude that $U_{[0, 2], n}$ will successfully realize $(\frac{x_0}{2^n}, \frac{x_1}{2^n}, \frac{x_2}{2^n})$.
\end{proof}

\begin{figure}[!t]
\centering
\includegraphics[width=3.0in]{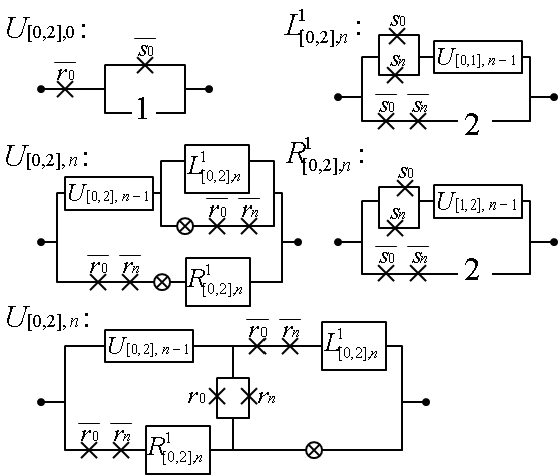}
\caption{\textbf{3-state UPG reduced} After reducing the exponential UPG, we get two linear UPGs. One is a sp-circuit and uses 2 stochastic switches. The other is non-sp and uses 1 stochastic switch.}
\label{UPG3RED}
\end{figure}
\begin{lemma}
A 3-state UPG $U_{[0, 2], n}$ with inputs $r_0, r_n, \ldots, r_2, r_1$ and $s_0, s_n, \ldots, s_2, s_1$ can be constructed with a quadratic number of switches using either of the two recursive constructions in Figure \ref{UPG3RED}, where the bits used in $U_{[0, 1], n-1}$ are $r_0, r_{n-1}, \ldots, r_2, r_1$, the bits used in $U_{[1, 2], n-1}$ are $s_0, s_{n-1}, \ldots, s_2, s_1$, and the bits used in $U_{[0, 2], n-1}$ are $r_0, r_{n-1}, \ldots, r_2, r_1$ and $s_0, s_{n-1}, \ldots, s_2, s_1$.
\end{lemma}
\begin{proof}
We will prove this by algebraically reducing the exponential UPG. Let $p_i = (\frac{1}{2}, 0, \frac{1}{2})$ be i.i.d. We also introduce some new variables to make the final expressions concise.
\begin{align*}
L_{[0, 2], n}^1 &= s_{0n}U_{[0, 1], n-1} + \bar{s}_{0n}\cdot 2,\\
R_{[0, 2], n}^1 &= s_{0n}U_{[1, 2], n-1} + \bar{s}_{0n}\cdot 2
\end{align*}
A visual of these is in Figure \ref{UPG3RED}. $R_{[0, 2], n}^1$ is defined so that $R_{[0, 2], n} = r_{0n}U_{[0, 2], n-1} + \bar{r}_{0n}R_{[0, 2], n}^1$. $L_{[0, 2], n}^1$ is defined similarly except that the $U_{[0, 2], n-1}$ component is replaced by a $2$. We notice that,
\begin{align*}
L_{[0, 2], n} &= U_{[0, 2], n-1}L_{[0, 2], n}\\
              &= U_{[0, 2], n-1}(r_{0n}\cdot 0 + \bar{r}_{0n}L_{[0, 1], n-1}^1)
\end{align*}

\begin{align*}
U_{[0, 2], n} &= L_{[0, 2], n} + p_nR_{[0, 2],n}\\
              &= L_{[0, 2], n}U_{[0, 2], n-1} + p_nR_{[0, 2], n}\\
              &= (r_{0n}\cdot 0 + \bar{r}_{0n}L_{[0, 2], n}^1)U_{[0, 2], n-1}\\
              &\ \ \ \ + p_n(r_{0n}U_{[0, 2], n-1} + \bar{r}_{0n}R_{[0, 2], n}^1)\\
              &= \bar{r}_{0n}L_{[0, 2], n}^1U_{[0, 2], n-1} + p_n(r_{0n}U_{[0, 2], n-1} + \bar{r}_{0n}R_{[0, 2], n}^1)\\
              &= U_{[0, 2], n-1}(\bar{r}_{0n}L_{[0, 2], n}^1 + p_nr_{0n}) + p_n\bar{r}_{0n}R_{[0, 2], n}^1
\end{align*}
This is exactly the form of the series parallel construction in Figure \ref{UPG3RED}.
\begin{align*}
U_{[0, 2], n} &= L_{[0, 2], n} + p_nR_{[0, 2], n}\\
              &= L_{[0, 2], n}U_{[0, 2], n-1} + p_nR_{[0, 2], n}\\
              &= (r_{0n}\cdot 0 + \bar{r}_{0n}L_{[0, 2], n}^1)U_{[0, 2], n-1}\\
              &\ \ \ \ + p_n(r_{0n}U_{[0, 2], n-1} + \bar{r}_{0n}R_{[0, 2], n}^1)\\
              &= \bar{r}_{0n}L_{[0, 2], n}^1U_{[0, 2], n-1} + p_n(r_{0n}U_{[0, 2], n-1} + \bar{r}_{0n}R_{[0, 2], n}^1)\\
              &= [U_{[0, 2], n-1}][\bar{r}_{0n}L_{[0, 2], n}^1] + [\bar{r}_{0n}R_{[0, 2], n}^1][r_{0n}][\bar{r}_{0n}L_{[0, 2], n}^1]\\
              &\ \ \ \ + [U_{[0, 2], n-1}][r_{0n}][p_n] + [\bar{r}_{0n}R_{[0, 2], n}^1][p_n]
\end{align*}
This is exactly the form of the non-series parallel construction in Figure \ref{UPG3RED}.
\end{proof}
Finally, as in the 2-state UPG, we will remove the repetitive instances of $r_0$ and $s_0$.

\begin{figure}[!t]
\centering
\includegraphics[width=3.0in]{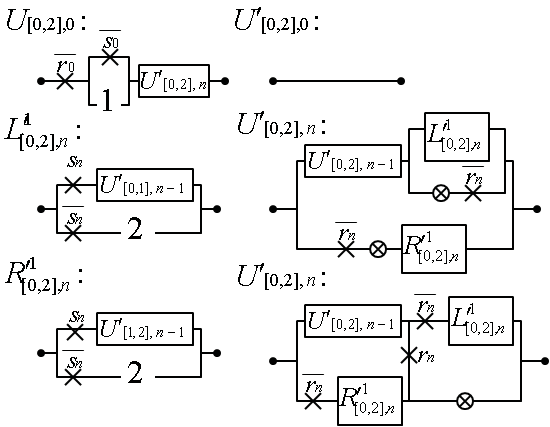}
\caption{\textbf{3-state UPG removed bit} Even though $r_0$ and $s_0$ were important in controlling the cases of our algorithm, it turns out we can remove it from the recursive construction and just append it at the end.}
\label{UPG3NOBIT}
\end{figure}
\begin{figure}[!t]
\centering
\includegraphics[width=3.0in]{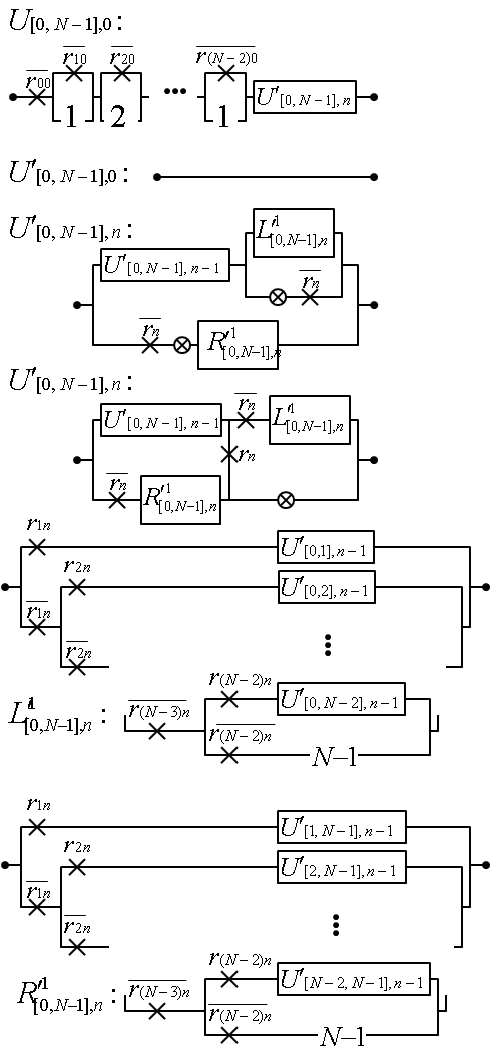}
\caption{\textbf{$N$-state UPG} The $N$-state UPG is almost exactly the same as the 3-state UPG.}
\label{UPGNNOBIT}
\end{figure}

\begin{theorem}
A 3-state UPG $U_{[0, 2], n}$ with inputs $r_0, r_n, \ldots, r_2, r_1$ and $s_0, s_n, \ldots, s_2, s_1$ can be constructed with a quadratic number of switches using either of the two recursive constructions in Figure \ref{UPG3NOBIT}, where the bits used in $U_{[0, 1], n-1}$ are $r_0, r_{n-1}, \ldots, r_2, r_1$, the bits used in $U_{[1, 2], n-1}$ are $s_0, s_{n-1}, \ldots, s_2, s_1$, and the bits used in $U_{[0, 2], n-1}$ are $r_0, r_{n-1}, \ldots, r_2, r_1$ and $s_0, s_{n-1}, \ldots, s_2, s_1$.
\end{theorem}
\begin{proof}
We will prove this by induction. Assume $U_{[0, 2], n-1} = \bar{r}_0(\bar{s}_0 + 1)U'_{[0, 2], n-1}$. Then, by breaking up the circuit into cases based on $r_0$ and $s_0$, we will show that $U_{[0, 2], n} = \bar{r}_0(\bar{s}_0 + 1)U'_{[0, 2], n}$. 

Case 0: $r_0 = 2, s_0 = 2$.
In this case, we know that 
\begin{align*}
U_{[0, 2], n} &= 0\\
\bar{r}_0 &= 0 \implies \bar{r}_0(\bar{s}_0 + 1)U'_{0, 2], n} = 0
\end{align*}
Therefore, $U_{[0, 2], n} = \bar{r}_0(\bar{s}_0 + 1)U'_{[0, 2], n}$.

Case 1: $r_0 = 0, s_0 = 2$.
\begin{align*}
U_{[0, 2], n}  &= L_{[0, 2], n} + p_nR_{[0, 2], n}\\
               &= L_{[0, 2], n}U_{[0, 2], n-1} + p_nR_{[0, 2], n}\\
U'_{[0, 2], n} &= L'_{[0, 2], n} + p_nR'_{[0, 2], n}\\
               &= L'_{[0, 2], n}U'_{[0, 2], n-1} + p_nR'_{[0, 2], n}
\end{align*}
We want to show that $L_{[0, 2], n}U_{[0, 2], n-1} = \bar{r}_0(\bar{s}_0 + 1)L'_{[0, 2], n}U'_{[0, 2], n-1}$, $R_{[0, 2], n} = \bar{r}_0(\bar{s}_0 + 1)R'_{[0, 2], n}$.
\begin{align*}
\bar{r}_0(\bar{s}_0 + 1) &= 1\\
L'_{[0, 2], n}U'_{[0, 2], n-1}   &= \bar{r}_nU'_{[0, 1], n-1}\\
L_{[0, 2], n}U_{[0, 2], n-1}     &= \bar{r}_nU_{[0,1], n-1}\cdot 1\cdot U_{[0, 2], n-1}\\
                                 &= \bar{r}_nU_{[0, 1], n-1}\cdot 1\\
                                 &= \bar{r}_n\cdot 1 \cdot U'_{[0, 1], n-1}\cdot 1\\
                                 &= \bar{r}_0(\bar{s}_0 + 1)L'_{[0, 2], n}U'_{[0, 2], n-1}\\
R'_{[0, 2], n}                   &= r_nU'_{[0, 2], n-1} + \bar{r}_n\\
R_{[0, 2], n}                    &= r_nU_{[0, 2], n-1} + \bar{r}_n(U_{[1, 2], n-1} + 1)\\
                                 &= r_nU_{[0, 2], n-1}\cdot 1 + \bar{r}_n\cdot 1\\
                                 &= r_n\cdot 1 \cdot U'_{[0, 2], n-1} \cdot 1 + \bar{r}_n \cdot 1\\
                                 &= \bar{r}_0(\bar{s}_0 + 1)R'_{[0, 2], n}
\end{align*}
Therefore, $U_{[0, 2], n} = \bar{r}_0(\bar{s}_0 + 1)U'_{[0, 2], n}$.

Case 2: $r_0 = 0, s_0 = 0$.
In this case, we know that
\begin{align*}
r_{0n} &= r_n\\
s_{0n} &= r_n\\
\bar{r}_0(\bar{s}_0 + 1) &= 2
\end{align*}
Therefore, $\bar{r}_0(\bar{s}_0 + 1)U'_{[0, 2], n} = U'_{[0, 2], n} = U_{[0, 2], n}$.
\end{proof}

Finally, as in the 2-state case, the 3-state UPG $U_{[i, i+2], n}$ will generate the appropriate distributions over the states $i, i+1, i+2$.

\subsection{$N$-state UPG}

The results will now be generalized to $N$-states. The steps and the proof are all identical to the 3-state version except that the variables $L$ and $R$ are defined to be generalized constructions. Since they are identical, we will define the $N$-state UPG and then jump to the conclusion and circuit diagrams.

{\it Definition ($N$-state UPG):}
A $N$-state UPG $U_{[0, N-1], n}$ is a circuit that generates distributions of the form $(\frac{x_0}{2^n}, \frac{x_1}{2^n}, ..., \frac{x_{N-1}}{2^n})$ using $N-1$ input vectors $\mathbf{r_0}, \mathbf{r_1}$, ..., $\mathbf{r_{N-2}}$, where each vector $\mathbf{r_i}$ has $n+1$ bits $r_{i0}$, $r_{i1}$, ..., $r_{in}$. When the input vectors $\mathbf{r_i} = (r_{i0}, r_{in}, r_{i(n-1)}, \ldots, r_{i2}, r_{i1})$ are set to the binary representation of $\frac{x_0 + x_1 + ... + x_i}{2^n}$ with the symbol $N-1$ replacing the boolean 1, then the circuit $U_{[0, N-1], n}$ will realize distribution $(\frac{x_0}{2^n}, \frac{x_1}{2^n}, ..., \frac{x_{N-1}}{2^n})$. In other words, to generate any desired binary distribution $(\frac{x_0}{2^n}, \frac{x_1}{2^n}, ..., \frac{x_{N-1}}{2^n})$, we set the input vector $\mathbf{r_0}$ to $\frac{x_0}{2^n}$, the input vector $\mathbf{r_1}$ to the sum $\frac{x_0 + x_1}{2^n}, \ldots$, and the input vector $\mathbf{r_{N-2}}$ to the sum $\frac{x_0 + x_1 + ... + x_{N-2}}{2^n}$.

\begin{theorem}
A $N$-state UPG $U_{[0, N-1], n}$ with inputs $\mathbf{r_i}, 0 \leq i \leq N-2$ can be constructed with a polynomial number of switches $O(n^{N-1})$ using either of the two recursive constructions in Figure \ref{UPGNNOBIT}, where the bits used in $U_{[0, i], n-1}$ are $\mathbf{r_0}, \mathbf{r_1}, \ldots, \mathbf{r_{i-1}}$ and the bits used in $U_{[i, N-1], n-1}$ are $\mathbf{r_i}, \mathbf{r_{i+1}}, \ldots, \mathbf{r_{N-2}}$.
\end{theorem}

\section{Partial Orders}

All current work has been done on states that are a total order. i.e. $0 < 1 < ... < n-1$. It is interesting to think about a logic on states that are in a partial order. We perform a cursory examination of constructing probabilities on lattices that are partial orders and find that we cannot generate many distributions on partial orders.

Composition rules must be generalized again since you cannot take a max or min of incomparable states. Instead of max and min, we use the $\vee$ (join) and $\wedge$ (meet) operators respectively.

\begin{theorem}[Partial Order Inexpressibility]
For the lattice in Figure \ref{lattice}b, no distributions of the form $v = (0, 1-p, p, 0)$, where $p \geq 0$ are realizable by building a sp circuit with any switch set unless $(0, 1-p, p, 0)$ itself is in the switch set.
\end{theorem}
\begin{proof}
Assume that $v$ is realizable by a series composition. That is, $v = xy$. Then,
\begin{align*}
\displaystyle v_{00} &= x_{00} + y_{00} - x_{00}y_{00} + x_{01}y_{10} + x_{10}y_{01}\\
                     &\Rightarrow x_{00} = 0 \text{ and } y_{00} = 0\\
              v_{11} &= 0 \\
                     &\Rightarrow x_{11} = 0 \text{ or } x_{11} = 0
\end{align*}
WLOG, we let $x_{11} = 0$, so that,
\begin{align*}
\displaystyle v_{10} &= x_{10}y_{11} + x_{11}y_{10} + x_{10}y_{10}\\
                     &= x_{10}y_{11} + x_{10}y_{10}\\
                     &= x_{10}(y_{11} + y_{10}) = p\\
                     &\Rightarrow x_{10} \neq 0\\
              v_{01} &= x_{01}y_{11} + x_{11}y_{01} + x_{01}y_{01}\\
                     &= x_{01}y_{11} + x_{01}y_{01}\\
                     &= x_{01}(y_{11} + y_{01}) = 1-p\\
                     &\Rightarrow x_{01} \neq 0\\
              v_{00} &= x_{01}y_{10} + x_{10}y_{01}\\
                     &\Rightarrow y_{10} = 0 \text{ and } y_{01} = 0\\
                     &\Rightarrow y_{11} = 1\\
                     &\Rightarrow x = v
\end{align*}
But then we have $x$ = $v$, so we have shown that need without $v$ itself, we cannot realize $v$ through a series connection. In the same way, we can show the same result for a parallel composition (left out since it is almost identical), which means that $v$ is not realizable with a sp circuit.
\end{proof}

\begin{figure}[!t]
\centering
\includegraphics[width=3.0in]{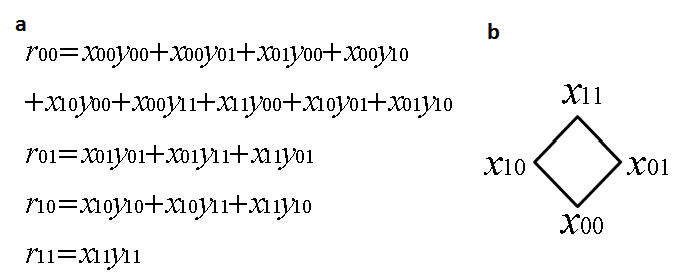}
\caption{\textbf{A simple partially ordered lattice.} This is an example of a simple partial order. \textit{(a)} An example of a series composition (meet) of switches $x$ and $y$. \textit{(b)} The lattice structure.}
\label{lattice}
\end{figure}

\section{Applications}
One intuitive understanding of this multivalued alphabet is to think of timings and dependencies. Switches represent events, and the states represent the discrete time when the event occurs. When two switches $x$ and $y$ are composed to make the circuit $z$, we are saying that for event $z$ to occur, it depends on events $x$ and $y$. If event $z$ requires both $x$ and $y$ to occur before it can occur, then this can be represented by a parallel connection since the time $z$ occurs will be the $\max(x, y)$. If event $z$ only needs either $x$ or $y$ to occur, then this can be represented by a series connection since the time $z$ occurs will be the $\min(x, y)$.

It is possible to implement multivalued relay circuits using physical (open and closed) relay switches. If we consider switches as normally closed and the state as the time it is opened, then we have a direct correspondence to the description above. When switches are composed in series, the time that the entire circuit is open, is when either of the switches open (min) and when the switches are composed in parallel, the time that the entire circuit is open is when both the switches are open (max). 

We can imagine a biological neural circuit working in a very similar way. Consider a network of neurons where each neuron only has incoming signals from 2 other neurons. Then, consider a neuron $x$ when incoming neurons $y$ and $z$. Neuron $x$ will fires when the incoming signals exceeds a certain threshold potential. If the threshold is low, we can imagine that neuron $x$ will fire once it receives the first incoming signal (i.e. $\min(x,y)$). If the threshold is high, we can imagine that neuron $x$ will only fire if it receives signals from both incoming neurons (i.e. $\max(x,y)$). In general, with $m$ incoming neuron signals, it becomes slightly more complex; depending on the threshold of the neuron of interest, the time it fires would be when it receives the 1st, 2nd, 3rd, ..., mth signal. These gates can be implemented by max/min gates.

There are a number of directions for extending this work to better model and understand biological circuits. These include studying joint distributions, stochastic functions, probability distributions on different state structures, and stochasticity in compositions. With further understanding of biological circuits, this can also aid in building artificial molecular circuits, such as in DNA computing\cite{Winfree}. In addition, stochastic circuits may be useful in modeling other stochastic networks in engineering.

\section{Conclusion}
In this paper, we studied probability generation in the context of multivalued relay circuits, a generalization of two-state relays. We proved a duality result on this max/min generalization and then proved construction algorithms for generating any rational probability distribution. We extended the robustness result to these algorithms and showed that switch error remains bounded linearly regardless of the circuit size. Finally, we constructed a universal probability generator for mapping deterministic inputs to probability distributions and demonstrated a basic non-realizability result for partial orders.

Further work on multivalued stochastic switches may have many biological applications such as for neural coding.

\section*{Acknowledgment}
The authors would like to thank Dan Wilhelm, Hongchao Zhou, and Ho-lin Chen for helpful discussions. They would also like to thank the Caltech SURF program, the Molecular Programming Project funded by the NSF Expeditions in Computing Program under grant CCF-0832824, and Aerospace Corporation for funding to make this research possible.

\end{document}